\documentclass[journal]{IEEEtran}   
\usepackage{multirow}
\usepackage{diagbox}
\usepackage{amssymb}
\usepackage{amsmath}
\usepackage{graphicx}
\usepackage{cite}
\usepackage{citesort}
\usepackage{subfigure}
\usepackage{graphicx,epstopdf}
\usepackage{epsfig}	
\usepackage{cite,graphicx,amsmath,amssymb}
\usepackage{comment}

\usepackage{amssymb}
\usepackage{amsmath}
\usepackage{cite}
\usepackage{url}
\usepackage{xcolor}
\usepackage{cite,graphicx,amsmath,amssymb}
\usepackage{subfigure}
\usepackage{citesort}
\usepackage{fancyhdr}
\usepackage{mdwmath}
\usepackage{mdwtab}
\usepackage{caption}
\usepackage{amsthm}
\usepackage{lipsum}

\newtheorem{theorem}{Theorem}

\newtheorem{lemma}{Lemma}

\newtheorem{corollary}{Corollary}

\newtheorem{remark}{Remark}  
\newtheorem{proposition}{Proposition}

\usepackage [latin1]{inputenc}
\UseRawInputEncoding  

\makeatletter
\def\ScaleIfNeeded{%
\ifdim\Gin@nat@width>\linewidth \linewidth \else \Gin@nat@width
\fi } \makeatother

\begin{document}

\title{\Huge{Simultaneously Transmitting and Reflecting Reconfigurable Intelligent Surface Assisted NOMA Networks}}

\author{Xinwei~Yue,~\IEEEmembership{Senior Member~IEEE}, Jin Xie, Yuanwei\ Liu,~\IEEEmembership{Senior Member~IEEE}, Zhihao~Han,  Rongke\ Liu,~\IEEEmembership{Senior Member~IEEE} and  Zhiguo~Ding,~\IEEEmembership{Fellow, IEEE}

\thanks{X. Yue and J. Xie are with the Key Laboratory of Information and Communication Systems, Ministry of Information Industry and also with the Key Laboratory of Modern Measurement $\&$ Control Technology, Ministry of Education, Beijing Information Science and Technology University, Beijing 100101, China (email: \{xinwei.yue, jin.xie\}@bistu.edu.cn).}
\thanks{Y. Liu is with the School of Electronic Engineering and Computer Science, Queen Mary University of London, London E1 4NS, U.K. (email: yuanwei.liu@qmul.ac.uk).}
\thanks{Z. Han and R. Liu are with the School of Electronic and Information Engineering, Beihang University, Beijing 100191, China (email: \{hzh$\_$95, rongke$\_$liu\}@buaa.edu.cn).}
\thanks{Z. Ding is with the Department of Electrical Engineering, Princeton University, Princeton, USA and also with the School of Electrical and Electronic Engineering, the University of Manchester, Manchester, U.K. (e-mail: zhiguo.ding@manchester.ac.uk).}
}

\maketitle

\begin{abstract}
Simultaneously transmitting/refracting and reflecting reconfigurable intelligent surface (STAR-RIS) has been introduced to achieve full coverage area. This paper investigate the performance of STAR-RIS assisted non-orthogonal multiple access (NOMA) networks over Rician fading channels, where the incidence signals sent by base station are reflected and transmitted to the nearby user and distant user, respectively. To evaluate the performance of STAR-RIS-NOMA networks, we derive new approximate expressions of outage probability and ergodic rate for a pair of users, in which the imperfect successive interference cancellation (ipSIC) and perfect SIC (pSIC) schemes are taken into consideration. Based on the asymptotic expressions, the diversity orders of the nearby user with ipSIC/pSIC and distant user are achieved carefully. The high signal-to-noise ratio slopes of ergodic rates for nearby user with pSIC and distant user are equal to  $one$ and $zero$, respectively. In addition, the system throughput of STAR-RIS-NOMA is discussed in delay-limited and delay-tolerant modes. Simulation results are provided to verify the accuracy of the theoretical analyses and demonstrate that: 1)  The outage probability of STAR-RIS-NOMA outperforms that of  STAR-RIS assisted orthogonal multiple access (OMA) and conventional cooperative communication systems; 2) With the increasing of reflecting elements $K$  and Rician factor $\kappa $, the STAR-RIS-NOMA networks are capable of attaining the enhanced performance; and 3) The ergodic rates of STAR-RIS-NOMA are superior to that of STAR-RIS-OMA.
\end{abstract}

\begin{keywords}
Non-orthogonal multiple access, reconfigurable Intelligent surface, simultaneous transmitting and reflecting, outage probability, ergodic rate.
\end{keywords}

\section{Introduction}
With the commercial deployment of the fifth-generation (5G) communication systems, the key technologies of physical layer have begun to be pre-researched for sixth-generation (6G) communication networks. Compared with 5G systems, the aims of 6G networks are to meet the requirements of dynamical businesses and provide the enhanced spectral/energy efficiency, global coverage and better intelligence levels, etc \cite{YouTowards6G}. Non-orthogonal multiple access (NOMA) with the characteristics of high spectrum efficiency and supporting giant connections has been viewed as a promising multiple access candidate for 6G networks \cite{Khan2020NOMA}. As clearly expressed that the NOMA scheme has ability to achieve the boundary of capacity region and better fairness with respect to orthogonal multiple access (OMA) \cite{Tse2005,Ding2017Mag}, where both superposition coding and successive interference cancellation (SIC) are employed at the transmitters and receivers, respectively.

Integrating NOMA with other physical layer techniques have been discussed extensively based on different application scenarios and requirements \cite{Yuan2021NOMA6G,Yuanwei2021NOMA}. The concept of cooperative NOMA communications was proposed in \cite{Ding2014Cooperative}, where the users with better conditions are selected as relays to guarantee the service quality of cell edge users. Furthermore, the authors of  \cite{Zhang2017FDNOMA,Yue8026173} investigated the performance of full/half-duplex (FD/HD) cooperative NOMA systems in terms of outage probability and ergodic rate. With the emphasis on secure communications, the security performance of NOMA networks was evaluated in \cite{Yue2019UnifiedPLS}, where the internal and  external eavesdropping scenarios are examined carefully. Applying NOMA to random access, the author in \cite{Rana2019FreeNOMA} pointed out that the NOMA based grant-free scheme can potentially support giant connectivity by considering latency and reliability. Additionally, the NOMA assisted semi-grant-free transmission scheme was analyzed in \cite{Ding2019Semi}, which is regarded as a compromise between grant-free transmission and grant based schemes.
To enhance the spectrum usage of an unmanned aerial vehicle (UAV), the authors of \cite{Li2020UAVNOMA} surveyed the performance of coverage probability and achievable rate for UAV-enabled NOMA communications. In \cite{Yue2020SatelliteNOMA}, the authors introduced the use of NOMA to satellite networks and evaluated the outage behaviors of terrestrial users with order statistics. Recently, the backscatter communications aided NOMA networks were studied in  \cite{Ding2021BACNOMA}, which has the ability to effectively support ultra-massive machine type scenarios.

Simultaneously transmitting/refracting and reflecting reconfigurable intelligent surface (STAR-RIS)  has given rise to the heated discussions in both academia and industry communities \cite{Liu2021STAR360,Mu2021STARWireless,Song2021Reflective}. The pivotal thought of STAR-RIS is to refract and reflect incident signals towards the users on the same side and opposite side of source, respectively. The employment of STAR-RIS has the following superiorities relative to reflecting-only RIS: i) STAR-RIS is capable of providing the flexible degree-of-freedom to propagate signals; ii) The coverage of STAR-RIS can be extended to the entire space ; and iii) STAR-RIS is usually designed to be optically transparent, which can be compatible with the current building structures. From the perspectives of hardware and system design, the authors of \cite{Liu2021STAR360} introduced the fundamental signal model of STAR-RIS, where the practical protocols of energy splitting, mode switching and time switching are proposed carefully. On the basis of these, the authors in \cite{Liu2021STAR} compared the differences between the conventional reflective-only RIS and STAR-RIS. To meet diverse requirements, the authors of \cite{Song2021Reflective} studied the system capacity of three types of RIS i.e., reflective, transmissive and hybrid types.
In \cite{Basar2021RIS}, the authors focused practicality on active, transmitter and transmissive-reflective types of RIS by evaluating their advantages and weaknesses relative to reflective RIS designs. As a further advance, the related works on RIS and RIS-NOMA  are surveyed exhaustively in the following two paragraphs.

\subsubsection{Related works on RIS}
The RIS-assisted wireless communications have sparked a lot of attention \cite{Hu2018Surfaces,Linglong2020RIS,Pan2021RIS6G}. In \cite{Liu2020RISOpportunities}, the authors researched the channel performance of RIS-assisted networks by categorizing the RIS illuminated space. With the goal of maximizing energy efficiency, the authors of \cite{Huang2019RIS} introduced the low complexity approaches by jointly designing both the transmit power allocation and phase shifts of reflecting elements. In \cite{Chien2021RISCorrelated}, a statistical descriptions of outage probability, ergodic rate  and symbol error rate were outlined  for RIS-assisted wireless communications over Rayleigh fading channels. Under Nakagami-$m$ channels conditions, the authors of \cite{Ibrahim2021RISCoverageNakagami} analyzed the performance of  coverage probability for RIS-assisted communication systems by exploiting moment generation functions. To shed light on the impact of the line-of-sight (LoS) component, the authors in \cite{Zhong9146875IRSRicanfading,Salhab2021IRSRician} evaluated the outage performance, ergodic capacity and average symbol error probability of RIS-assisted wireless works over Rician fading channels. Except the above contributions, several application scenarios i.e., applying RIS to massive device-to-device communications and facilitating simultaneous wireless and power transfer were highlighted in \cite{WuTowards2019}. Two-way communications between users aided by RIS were surveyed in \cite{Pradhan2020RIStwoway}, where the reciprocal or non-reciprocal channels are taken into account. From the viewpoint of security, the authors of \cite{Yang2020RISSecrecy} studied the secrecy outage behaviors and average secrecy capacity of RIS-assisted networks by using stochastic geometry.

\subsubsection{Related works on RIS-NOMA}
Until now, the RIS-assisted NOMA networks have been discussed from the perspective of the performance analyses  \cite{Ding2020IRSDesign,Yue2020IRSNOMA,Ding2020Shifting}. A simple design of RIS-NOMA transmission scheme was proposed in \cite{Ding2020IRSDesign}, where the increasing number of reflection elements can effectively reduce the outage probability. Inspired by this work, the authors of \cite{Yue2020IRSNOMA} investigated the outage probability, ergodic rate and energy efficiency of RIS-NOMA with perfect SIC (pSIC) and imperfect SIC (ipSIC). In \cite{Ding2020Shifting}, the impact of coherent/random phase shifting on the outage performance was examined for RIS-NOMA networks.
As a further advance, the outage probability and ergodic rate of prioritized user for RIS-NOMA were studied in \cite{Hou2019IRS} by designing the passive beamforming weights.
Given the users' rate, the authors of \cite{ZhengIRSUserpairing,Hong2020RISNOMAPower} surveyed the transmit power minimization problems with discrete phase shifts for RIS-aided NOMA and OMA.
Aims to mitigate the transmission interference, the authors made the use of a novel NOMA solution with RIS partitioning \cite{Basar2020IRSNOMA}, where the fairness among users can be maximized in detail. In  \cite{Lina2020IRSNOMASEP}, the pairwise error probability and phase shift designing for RIS-NOMA networks were investigated by employing the ipSIC and group-based SIC schemes.
According to whether there is a direct link between the base station (BS) and users, the authors of \cite{Cheng2020MultipleIRSNOMA} analyzed the outage behaviors of multiple RISs-assisted NOMA networks with discrete phase shifting. Moreover, the ergodic rate performance of RIS-aided uplink and downlink NOMA networks was surveyed in \cite{Cheng2021UplinkdownlinkNOMARIS}, which revealed the superiority of the RIS over full-duplex decode-and-forward (DF) relaying. Additionally, the authors of \cite{Xidong2020RISNOMA} maximized the sum rate of RIS-NOMA networks by jointly optimizing the active at the BS and passive beamforming at the RIS. In \cite{Hong2021TwoCellIRS}, the phase shifting and power allocation of RIS-aided two-cell NOMA networks was studied by invoking the joint detection.

\subsection{Motivation and Contributions}
As previously mentioned above, the existing research contributions assume RIS to be operated in the reflection mode, where the destination is only located on the same side of source. This geographical restriction may not always be satisfied in practical applications, and also restraints the effectiveness and agility of RIS. However, the STAR-RIS can refract and reflect the incident signals to the users located at different sides of the surface, which is capable of supplying the full-space coverage of smart radio environments.
 To broaden the applications of STAR-RIS, the authors in \cite{Zhang2021STARNOMA} investigated the outage behaviors of users with pSIC for STAR-RIS assisted NOMA networks. This assumption of pSIC might not be valid at receiver in practical scenarios, since there still exist several potential implementation issues by using SIC (i.e., complexity scaling and error propagation).
To the best of our knowledge, the performance of STAR-RIS-NOMA with ipSIC/pSIC over Rician fading channels is not researched yet. More specifically, we investigate the performance of a pair of users i.e., the nearby user $n$ and distant user $m$ for STAR-RIS-NOMA networks in terms of outage probability and ergodic rate. The direct communication link from the BS to nearby user are taken into account in detail.
Additionally, the outage probability and ergodic rate of STAR-RIS-OMA are also evaluated seriously. According to the aforementioned explanations, the primary contributions of this manuscript are summarized as follows:
\begin{enumerate}
  \item We derive the approximate expressions of outage probability for user $n$ with ipSIC/pSIC and user $m$ over Rician fading channels. Based on Laplace transform and convolution theorem, we further calculate the asymptotic outage probability and then obtain the diversity orders of user $n$ with ipSIC/pSIC and user $m$. We observe that the diversity orders of user $n$ with pSIC and user $m$ are related to the configure elements. We also derive the approximate expressions of outage probability for STAR-RIS-OMA.
    \item  We compare the outage behaviors of user $n$ with ipSIC/pSIC and user $m$ for STAR-RIS-NOMA with STAR-RIS-OMA. We further confirm that the outage probability of  STAR-RIS-NOMA with pSIC is superior to that of STAR-RIS-OMA and conventional cooperative communication systems. As the reconfigurable elements $K$  and Rician factor $\kappa $ increases, the STAR-RIS-NOMA networks is able to achieve the enhanced outage performance.
      \item  We derive the asymptotic expressions of ergodic rate for user $n$ with pSIC and user $m$ in STAR-RIS-NOMA networks. An upper bound for ergodic rate of user $n$ with pSIC is provided to approximate the exact expression. Based on analytical results, the high signal-to-noise ratio (SNR) slopes of ergodic rate for user $n$ and user $m$ are achieved. We confirm that the ergodic rate of STAR-RIS-NOMA is superior to that of STAR-RIS-OMA.
         \item  We evaluate the system throughput of STAR-RIS-NOMA networks in both delay-limited and delay-tolerant modes. In delay-limited mode, the  system throughput of STAR-RIS-NOMA with pSIC are superior to that of STAR-RIS-OMA and conventional cooperative communication systems at high SNRs.  In delay-tolerant mode, the  system throughput of STAR-RIS-NOMA networks with pSIC outperforms that of STAR-RIS-NOMA with ipSIC and STAR-RIS-OMA.
\end{enumerate}

\subsection{Organization and Notations}
The remainder of this paper is organized as follows.   In Section \ref{Network Model}, the system model of STAR-RIS-NOMA networks is introduced. The outage behaviors of STAR-RIS-NOMA are evaluated in Section \ref{Outage Probability}. More specially, the approximate expressions of outage probability for user $n$ and user $m$ are provided. The ergodic rate of user $n$ and user $m$ is evaluated in Section \ref{Ergodic Rate}. Simulation results and discusses are presented in Section \ref{Numerical Results}, followed by concluding commentaries in Section \ref{Conclusion}. The proofs of mathematics are collected in the Appendix.

The main notations in this paper used are shown as follows. The probability density function (PDF) and cumulative distribution function (CDF) of a random variable $X$ are denoted by ${f_X}\left(  \cdot  \right)$ and ${F_X}\left(  \cdot  \right)$, respectively; $\mathbb{E}\{\cdot\}$ and $\mathbb{D}\{\cdot\}$  denotes the expectation and variance operations, respectively; The superscript ${\left(  \cdot  \right)^H}$ stands for  conjugate-transpose operation.

\section{System Model}\label{Network Model}
We consider a STAR-RIS assisted downlink NOMA network as shown in Fig. \ref{System_Model_Refracting_IRS_NOMA}, where the superposed signals are reflected and transmitted  to a pair of types' non-orthogonal users\footnote{It is worth noting that estimating multiple user pairs scenarios in STAR-RIS-NOMA networks can further enrich the contents of the paper considered, which will be set aside in our future work.}, i.e., the nearby user $n$ and distant user $m$ simultaneously. Due to the serious blockage and complicated wireless environment, we assume that the link from the BS to user $m$ is not available or even fall into complete outage status. More specifically, the user $n$ is on the side of base station (BS) in comparison to STAR-RIS, which has ability to receive both the signal from the BS and signal reflected by SRAR-RIS.
The user $m$ is located on the other side of STAR-RIS, which has only ability to receive the signals transmitted by STAR-RIS. The BS and users are equipped with single antenna, and the STAR-RIS consists of $2K$ configurable elements. We assume that the STAR-RIS elements are divided into two groups, where the first group of STAR-RIS elements are employed to fully reflect signals for reflecting links and the other group of STAR-RIS elements carries out the full refraction mode in the transmitting links. In actual, $K$ elements are exploited for the reflecting links and the remaining elements are used for transmitting links. We denote the complex channel coefficients from the BS to user  $n$, from BS to STAR-RIS, and then from STAR-RIS  to user $\varphi$ by ${{h_{sn}}}$, ${{\bf{h}}_{sr}} \in \mathbb{C}{^{K \times 1}}$ and ${{{\bf{h}}_{r\varphi} }} \in \mathbb{C}{^{K \times 1}}$ with $\varphi  \in \left\{ {n,m} \right\}$, respectively.
The wireless communication links for STAR-RIS network are modeled as the Rician fading channels. The effective cascade channel gains from the BS to STAR-RIS, and then to user $n$ and user $m$ can be written as  ${{\bf{h}}_{rm}^H{{\bf {\Theta}} _R}{{\bf{h}}_{sr}}}$ and ${{\bf{h}}_{rn}^H{{\bf {\Theta}} _T}{{\bf{h}}_{sr}}}$, respectively, where ${{\bf {\Theta}} _R} = {\rm{diag}}( {\sqrt {\beta _1^r} {e^{j\theta _1^r}}, ... ,\sqrt {\beta _k^r} {e^{j\theta _k^r}}, ... ,\sqrt {\beta _K^r} {e^{j\theta _K^r}}} )$ and ${{\bf {\Theta}} _T} = {\rm{diag}}( {\sqrt {\beta _1^t} {e^{j\theta _1^t}}, ... ,\sqrt {\beta _k^t} {e^{j\theta _k^t}}, ... ,\sqrt {\beta _K^t} {e^{j\theta _K^t}}} )$ denote the reflecting and  transmitting phase shifting matrixes of the STAR-RIS, respectively. $\sqrt {{\beta_k^r}}$, $\sqrt {{\beta_k^t}}  \in \left[ {0,1 } \right]$ and ${\theta _k^r}$, ${\theta _k^t} \in \left[ {0,2\pi } \right)$ denote the energy coefficient and phase shift of the $k$-th element for reflecting and transmitting responses, respectively. Applying mode switching protocol stated in \cite{Liu2021STAR360}, the $K$ configurable elements for reflecting links, we have ${\beta _k^r} = 1$ and $ {\beta _k^t} = 0$, while for the remaining $K$ elements for transmitting links, we have ${\beta _k^r} = 0$ and $ {\beta _k^t} = 1$. To support the requirements of diverse scenarios, the elements of STAR-RIS can be operated in full transmission mode or full reflection mode by adjusting the amplitude coefficients for transmission and reflection, which can be seen as the special case of STAR-RIS. For full reflection mode, each elements only reflect the incident signals from BS due to the copper backplane. On the contrary, the incident signals only penetrate the elements without the copper backplane for full transmission mode. Since the channel estimation and feedback process are not the consideration of this manuscript, we assume that perfect channel state information (CSI) at the BS and perfect feedback information to the STAR-RIS can be achieved carefully. Our future work will relax this idealized assumption and more details regarding the channel estimation of RIS can be found in \cite{Pan9180053,Zappone2019Overhead}.

\begin{figure}[t!]
\centering
 \includegraphics[width= 3.4in, height=1.8in]{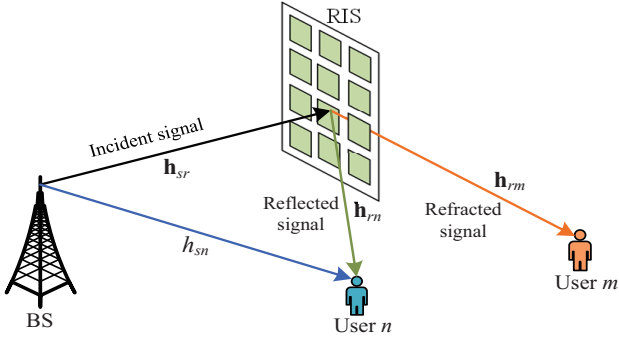}
 \caption{System model of STAR-RIS assisted NOMA networks.}
\label{System_Model_Refracting_IRS_NOMA}
\end{figure}

\subsection{Signal Model}
According to the principle of superposition coding \cite{Tse2005}, the BS broadcasts the superposed signals to a pair of users by the virtue of STAR-RIS.
In addition to receiving the signal from the BS, user $n$ also receives the signal reflected from the STAR-RIS. At this moment, the observation at user $n$ can be written as
\begin{align}\label{The received signals of user n}
{y_n} = \left( {{h_{sn}} + {\bf{h}}_{rn}^H{{\bf {\Theta}} _R}{{\bf{h}}_{sr}}} \right)\left( {\sqrt {{a_n}{P_s}} {x_n} + \sqrt {{a_m}{P_s}} {x_m}} \right) + {{\tilde n}_n},
\end{align}
where $x_{\varphi }$ denotes the unity power signal for user ${\varphi }$. The corresponding power allocation factors $a_n$ and $a_m$ satisfy the relationship ${a_n} < {a_m}$ with ${a_n}+{a_m} = 1$, which is for the viewpoint of user fairness. It is worth noting that the fixed power allocation among users is taken into consideration for STAR-RIS NOMA networks. $P_{s}$ denotes the normalized transmission power at BS. ${{\tilde n}_n}$ is the additive white Gaussian noise (AWGN) with mean power $N_{0}$ at user $n$. ${h_{sn}} = \sqrt {{\alpha _{sn}}} \left( {\sqrt {\frac{\kappa }{{\kappa  + 1}}}  + \sqrt {\frac{1}{{\kappa  + 1}}} {{\tilde h}_{sn}}} \right)$ denotes the channel coefficient from the BS to user $n$ and ${{\tilde h}_{sn}} \sim {\cal C}{\cal N} \left( {0,1} \right)$, where ${\alpha _{sn}} = d_{sn}^{ - \alpha }$ and $d_{sn}$ is the distance from BS to user $n$, {$\alpha $ is the path loss exponent.} 
The Rician factor is denoted by $\kappa $ and when $\kappa $ is set to zero, the corresponding Rician fading channel will be reduced into the Rayleigh fading channels. When $\kappa $ tends to infinity, the corresponding channel only exist the fixed line of sight component. Denoting ${{\bf{h}}_{sr}} = {\left[ {h_{sr}^1 \cdots h_{sr}^k \cdots h_{sr}^K} \right]^H}$ and ${{\bf{h}}_{rn}} = {\left[ {h_{rn}^1 \cdots h_{rn}^k \cdots h_{rn}^K} \right]^H}$, where $h_{sr}^k = \sqrt {\alpha _{sr}^k} \left( {\sqrt {\frac{\kappa }{{\kappa  + 1}}}  + \sqrt {\frac{1}{{\kappa  + 1}}} \tilde h_{sr}^k} \right) $ and $h_{rn}^k = \sqrt {\alpha _{rn}^k} \left( {\sqrt {\frac{\kappa }{{\kappa  + 1}}}  + \sqrt {\frac{1}{{\kappa  + 1}}} \tilde h_{rn}^k} \right)$ are the channel coefficients from BS to the $k$-th reflecting element of STAR-RIS, and then from the $k$-th reflecting element to user $n$, respectively. Define $\alpha _{sr}^k = d_{sr}^{ - \alpha }$, where $\alpha _{rn}^k = d_{rn}^{ - \alpha }$, $d_{sr}$ and $d_{rn}$ denote the distances from the BS to STAR-RIS, and then to user $n$.
The fading gains $\tilde h_{sr}^k$ and $\tilde h_{rn}^k$ are complex Gaussian distributed with zero mean and unit variance, i.e.,
$\tilde h_{sr}^k \sim {\cal C}{\cal N}\left( {0,1} \right)$ and $\tilde h_{rn}^k \sim {\cal C}{\cal N}\left( {0,1} \right)$.

For user $m$, it does not receive the signal from the BS and only receive the signal transmitted by STAR-RIS, which can be given by
\begin{align}\label{The received signals of user m}
{y_m} = {\bf{h}}_{rm}^H{{\bf {{\bf {\Theta}}}} _T}{{\bf{h}}_{sr}}\left( {\sqrt {{a_n}{P_s}} {x_n} + \sqrt {{a_m}{P_s}} {x_m}} \right) + {{\tilde n}_m},
\end{align}
where ${{\bf{h}}_{rm}} = {\left[ {h_{rm}^1 \cdots h_{rm}^k \cdots h_{rm}^K} \right]^H}$ and $h_{rm}^k = \sqrt {\alpha _{rm}^k} \left( {\sqrt {\frac{\kappa }{{\kappa  + 1}}}  + \sqrt {\frac{1}{{\kappa  + 1}}} \tilde h_{rm}^k} \right)$ denotes the channel coefficient from the $k$-th reflecting element of STAR-RIS to user $m$ with $\tilde h_{rm}^k \sim {\cal C}{\cal N}\left( {0,1} \right)$.
${{\tilde n}_m}$ is AWGN with mean power $N_{0}$ at user $m$.

In addition, user $n$ is on the side of BS relative to STAR-RIS, which performs SIC to first detect the signal $x_m$ of user $m$, then proceeding to subtract it and decode its signal. Hence the corresponding signal-plus-interference-to-noise ratio (SINR) can be given by
\begin{align}\label{The SINR of the n-th user to detect the m-th user}
{\gamma _{n \to m}} = \frac{{{{\left| {{h_{sn}} + {\bf{h}}_{rn}^H{{ \bf{\Theta }} _R}{{\bf{h}}_{sr}}} \right|}^2}\rho {a_m}}}{{{{\left| {{h_{sn}} + {\bf{h}}_{rn}^H{{ \bf{\Theta }} _R}{{\bf{h}}_{sr}}} \right|}^2}\rho {a_n}   + 1}},
\end{align}
where  $\rho  = \frac{{{P_s}}}{{{N_0}}}$ denotes the transmit SNR. After applying SIC technology, the SINR of user $n$, who needs to decode the information of itself is given by
\begin{align}\label{The SINR of the n-th user}
{\gamma _n} = \frac{{{{\left| {{h_{sn}} + {\bf{h}}_{rn}^H{{\bf {\Theta}} _R}{{\bf{h}}_{sr}}} \right|}^2}\rho {a_n}}}{{\varpi {{\left| {{h_I}} \right|}^2}\rho  + 1}},
\end{align}
where ${{h_I}}  \sim {\cal C}{\cal N}\left( {0,{\Omega _I}} \right)$ denotes the residual interference from ipSIC. More specifically, $\varpi {\rm{ = 0}}$ and $\varpi {\rm{ = 1}}$ denote the pSIC and ipSIC operations, respectively.

The SINR of user $m$ to decode its the information by treating the signal $x_n$ of user $n$ can be given by
\begin{align}\label{The SINR of the m-th user}
{\gamma _m} = \frac{{{{\left| {{\bf{h}}_{rm}^H{{\bf {\Theta}} _T}{{\bf{h}}_{sr}}} \right|}^2}\rho {a_m}}}{{{{\left| {{\bf{h}}_{rn}^H{{\bf {\Theta}} _T}{{\bf{h}}_{sr}}} \right|}^2}\rho {a_n} + 1}}.
\end{align}
\subsection{STAR-RIS-OMA}
In this subsection, the STAR-RIS-OMA scheme is selected as one of a baseline for the purpose of comparison, where the RIS is deployed to assist the BS to send the information to user $n$ and user $m$. On the condition of the above assumptions, the detecting SNRs of user $n$ and user $m$ for STAR-RIS-OMA can be given by
\begin{align}\label{The SNR of STAR-OMA Usern}
\gamma _n^{OMA} = {\left| {{h_{sn}} + {\mathbf{h}}_{rn}^H{{\bf {\Theta}} _R}{{\mathbf{h}}_{sr}}} \right|^2}\rho  {a_n},
\end{align}
and
\begin{align}\label{The SNR of STAR-OMA Userm}
\gamma _m^{OMA} = {\left| {{\mathbf{h}}_{rm}^H{{\bf {\Theta}} _T}{{\mathbf{h}}_{sr}}} \right|^2}\rho {a_m},
\end{align}
respectively.

\subsection{Channel Statistical Properties}
In this subsection, the channel statistical properties of Rician and cascade Rician channels are provided, which will be employed to evaluate outage behaviors for STAR-RIS-NOMA networks in the following sections.

In light of the above discussions, the channel coefficient, i.e., $h_{sn}$ from the BS to user $n$ follows Rician distribution, where the PDF and CDF of ${\left| {{h_{sn}}} \right|}$ can be given by 
\begin{align}\label{The PDF of Rician channels}
{f_{\left| {{h_{sn}}} \right|}}\left( x \right) = \frac{{2x\left( {\kappa  + 1} \right)}}{{{\alpha _{sn}}{e^\kappa }}}{e^{ - \frac{{{x^2}\left( {\kappa  + 1} \right)}}{{{\alpha _{sn}}}}}}{I_0}\left( {2x\sqrt {\frac{{\kappa \left( {\kappa  + 1} \right)}}{{{\alpha _{sn}}}}} } \right),
\end{align}
and
\begin{align}\label{The CDF of Rician channels}
{F_{{\left| {{h_{sn}}} \right|}}}\left( x \right) = 1 - Q\left( {\sqrt {2\kappa } ,x\sqrt {\frac{{2\left( {\kappa  + 1} \right)}}{{{\alpha _{sn}}}}} } \right),
\end{align}
respectively, where ${I_0}\left(  \cdot  \right)$ is the modified Bessel function of the first kind with order $zero$ and
$Q\left( {a,b} \right) = \int_b^\infty  {x{e^{ - \frac{{{x^2} + {a^2}}}{2}}}{I_0}\left( {ax} \right)dx} $
denotes the generalized Marcum Q-function \cite{Cantrell1987,Shnidman1989}.

Until now, two types of phase shifting designs, i.e., coherent phase shifting and random phase shifting are taken into account \cite{Ding2020Shifting,Yue2020IRSNOMA}. In coherent phase shifting scheme, the phase shift of each reflecting and transmitting element is matched with the phases of its incoming and outgoing fading channels\footnote{It is worth pointing out that the coherent phase shifting needs to carry out the perfect phase adjustment, while the random phase shifting belongs to sub-optimal scheme can avoid the requirement of perfect phase adjustment and reduce the system overhead. }, where the superior performance of STAR-RIS-NOMA networks can be attained carefully. Therefore the coherent phase shifting scheme is selected to deal with the cascade Rician channels. As a further development, the PDF of cascade Rician fading channels from the BS to the $k$-th transmitting and reflecting element, and then to user $\varphi$, i.e., ${f_{\left| {h_{sr}^kh_{r\varphi }^k} \right|}}$  can be given by \cite{2006Probability}
\begin{align}\label{The CDF of cascade Rician channels}
{f_{\left| {h_{sr}^kh_{r\varphi }^k} \right|}}\left( x \right) = &  \frac{1}{{\sqrt {{\alpha _{sr}}{\alpha _{r\varphi}}} }}\sum\limits_{i = 0}^\infty  {\sum\limits_{j = 0}^\infty  {\frac{{4{x^{i + j + 1}}{{\left( {\kappa  + 1} \right)}^{i + j + 2}}{\kappa ^{i + j}}}}{{{{\left( {i!} \right)}^2}{{\left( {j!} \right)}^2}{e^{2\kappa }}}}} }\nonumber \\
& \times {\left( {{\alpha _{sr}}{\alpha _{r\varphi}}} \right)^{ - \frac{{i + l + 1}}{2}}}{K_{i - j}}\left[ {\frac{{2x\left( {\kappa  + 1} \right)}}{{\sqrt {{\alpha _{sr}}{\alpha _{r\varphi}}} }}} \right],
\end{align}
where ${K_v}\left(  \cdot  \right)$ is the modified Bessel function of the second kind with order $v$ \cite[Eq. (8.432)]{2000gradshteyn}.
For notational simplicity, we denote ${X_k} =\left| {h_{sr}^kh_{r\varphi}^k} \right|$. It can be observed that the mean $\mu _\varphi$ and variance $\Omega _\varphi $ of $X_k$ can be given by
\begin{align}\label{the mean of X_k}
\mu _\varphi= \mathbb{E}\left( {{X_k}} \right) = \frac{{\pi \sqrt {{\alpha _{sr}}{\alpha _{r\varphi}}} }}{{4\left( {\kappa  + 1} \right)}}{\left[ {{L_{\frac{1}{2}}}\left( { - \kappa } \right)} \right]^2},
\end{align}
and
\begin{align}\label{the variance of X_k}
\Omega _\varphi = \mathbb{D}\left( {{X_k}} \right) = {\alpha _{sr}}{\alpha _{rn}}\left\{ {1 - \frac{{{\pi ^2}}}{{16{{\left( {1 + \kappa } \right)}^2}}}{{\left[ {{L_{\frac{1}{2}}}\left( { - \kappa } \right)} \right]}^4}} \right\},
\end{align}
respectively, where ${L_{\frac{1}{2}}}\left(  \cdot  \right)$ is the Laguerre polynomial and can be denoted by ${L_{\frac{1}{2}}}\left( \kappa  \right) = {e^{\frac{1}{2}}}\left[ {\left( {1 - \kappa } \right){I_0}\left( { - \frac{\kappa }{2}} \right) - \kappa {I_1}\left( { - \frac{\kappa }{2}} \right)} \right]$.


\section{Outage Probability}\label{Outage Probability}
In this section, the performance of STAR-RIS-NOMA networks is investigated in terms of outage behaviors, where the approximate expressions of outage probability for user $n$ with ipSIC/pSIC and user $m$ are derived in detail. Based on these asymptotic expressions, we further provide the diversity orders of user $n$ with ipSIC/pSIC and user $m$, respectively.
\subsection{The Outage Probability of User $n$}
For the nearby user $n$, the SIC scheme is carried out to first detect the information of distant user $m$, and then decode its own signal. At this moment, the outage events can be explained as: 1) If user $n$ cannot detect the signal $x_m$ of user $m$, the interruption will arise; and 2) User $n$ has ability to decode $x_m$, while its own signal $x_n$ cannot decoded successfully.
With the help of these interruption events, the outage probability of user $n$ for STAR-RIS-NOMA networks can be approximated as
\begin{align}\label{the OP of user n with ipSIC}
{P_{n}} = {\rm{Pr}}\left( {{\gamma _{n \to m}} < {\gamma _{t{h_m}}}} \right) + {\rm{Pr}}\left( {{\gamma _{n \to m}} > {\gamma _{t{h_m}}},{\gamma _n} < {\gamma _{t{h_n}}}} \right),
\end{align}
where ${\gamma _{t{h_n}}} = {2^{{R_n}}} - 1$ and ${\gamma _{t{h_m}}} = {2^{{R_m}}} - 1$ denote the target SNRs of user $n$ and user $m$ with detecting the signals $x_n$ and $x_m$, respectively. $R_n$ and $R_m$ denotes the corresponding target rates.
As a further advance, the outage probability of user $n$ with ipSIC can be provided in the following theorem.

\begin{theorem}\label{Theorem1:the OP of user n with ipSIC under Rician fading channel}
Under Rician fading channels, the approximate expression for outage probability of user $n$ with ipSIC for STAR-RIS-NOMA networks is given by
\begin{align}\label{the OP of user n with ipSIC under Rician fading channel}
 &{P_{n,ipSIC}} \approx \Phi \sum\limits_{p = 1}^P {\sum\limits_{u = 1}^U {{H_p}{b_u}{\chi ^{{\varphi _n} + 1}}{{\left( {{x_u}{\rm{ + }}1} \right)}^{{\varphi _n}}}{e^{ - \frac{{\left( {{x_u}{\rm{ + }}1} \right)\chi }}{{2{\phi _n}}}}}} }  \nonumber \\
 & \times \left\{ {1 - Q\left( {\sqrt {2\kappa } ,\left[ {\chi  - \frac{{\left( {{x_u}{\rm{ + }}1} \right)\chi }}{2}} \right]\sqrt {\frac{{2\left( {\kappa  + 1} \right)}}{{{\alpha _{sn}}}}} } \right)} \right\} ,
\end{align}
where  $\varpi  = 1$, $\beta  = \frac{{{\gamma _{t{h_n}}}}}{{{a_n}\rho }}$, ${b_u} = \frac{\pi }{{2U}}\sqrt {1 - x_u^2} $, ${x_u} = \cos \left( {\frac{{2u - 1}}{{2U}}\pi } \right)$ ,
$\Phi  = \frac{1 }{{{2^{\varphi _n} }{\phi ^{{\varphi _n}  + 1}}\Gamma \left( {{\varphi _n}  + 1} \right)}}$, $\chi  = \sqrt {\beta \left( {\varpi {x_p}{\Omega _{LI}}\rho  + 1} \right)} $,  ${\varphi _n} = \frac{{\mu _n^2K}}{{{\Omega _n}}} - 1$, $\phi _{n} = \frac{{{\Omega _n}}}{{{\mu _n}}}$,
$\mu _n  = \frac{{\pi \sqrt {{\alpha _{sr}}{\alpha _{rn}}} }}{{4\left( {\kappa  + 1} \right)}}{\left[ {{L_{\frac{1}{2}}}\left( { - \kappa } \right)} \right]^2}$,  $\Omega _n  = {\alpha _{sr}}{\alpha _{rn}}\left\{ {1 - \frac{{{\pi ^2}}}{{16{{\left( {1 + \kappa } \right)}^2}}}{{\left[ {{L_{\frac{1}{2}}}\left( { - \kappa } \right)} \right]}^4}} \right\}$ and $\Gamma \left(  \cdot  \right)$ denotes the gamma function \cite[Eq. (8.310.1)]{2000gradshteyn}.
 ${{x_p}}$ and ${H_p}$ are the abscissas and weight of Gauss-Laguerre quadrature, respectively. In particularly, ${{{ x}_p}}$ is the $p$-th zero point of Laguerre polynomial ${{ L}_P}\left( {{{ x}_p}} \right)$ and the $p$-th weight can be expressed as ${H_p} = \frac{{{{\left( {P!} \right)}^2}{x_p}}}{{{{\left[ {{L_{P + 1}}\left( {{x_p}} \right)} \right]}^2}}}$. In addition, $P$ and $U$ are the parameters to guarantee a complexity-accuracy tradeoff.
\end{theorem}
\begin{proof}
See Appendix~A.
\end{proof}


\begin{corollary}\label{Corollary:the OP of  user n with pSIC under Rician fading channel}
For the special case with $\varpi=0$, the approximate expression for outage probability of user $n$ with pSIC for STAR-RIS-NOMA networks is given by
\begin{align}\label{the OP of the n-th user with pSIC under Rician fading channel}
& {P_{n,pSIC}} \approx {\sum\limits_{u = 1}^U {\frac{{\beta \left( {\kappa  + 1} \right){b_u}\left( {{x_u}{\rm{ + }}1} \right)}}{{{\alpha _{sn}}{e^\kappa }\Gamma \left( {{\varphi _n} + 1} \right)}}e} ^{ - \frac{{\beta \left( {\kappa  + 1} \right){{\left( {{x_u}{\rm{ + }}1} \right)}^2}}}{{4{\alpha _{sn}}}}}}  \nonumber  \\
&   \times {I_0}\left( {\left( {{x_u}{\rm{ + }}1} \right)\sqrt {\frac{{\beta \kappa \left( {\kappa  + 1} \right)}}{{{\alpha _{sn}}}}} } \right)\gamma \left( {{\varphi _n} + 1,\frac{{\left( {1 - {x_u}} \right)\sqrt \beta  }}{{2{\phi _n}}}} \right).
\end{align}
\begin{proof}
Upon substituting $\varpi=0$ into \eqref{the further OP expression of user n with ipSIC} and applying similar processes, the outage probability of user $n$ with pSIC can be written as
\begin{align}\label{OP J1 derived for D1}
 {P_{n,pSIC}} = &\frac{{2\left( {\kappa  + 1} \right)}}{{\Gamma \left( {{\varphi _n} + 1} \right){\alpha _{sn}}{e^\kappa }}}\int_0^{\sqrt \beta  } {x{e^{ - \left( {\kappa  + 1} \right)\frac{{{x^2}}}{{{\alpha _{sn}}}}}}}   \nonumber \\
  &  \times {I_0}\left( {2x\sqrt {\frac{{\kappa \left( {\kappa  + 1} \right)}}{{{\alpha _{sn}}}}} } \right)\gamma \left( {{\varphi _n} + 1,\frac{{\sqrt \beta   - x}}{{{\phi _n}}}} \right)dx .
\end{align}
By further applying Gaussian-Chebyshev quadrature into the above integral expression, we can obtain \eqref{the OP of the n-th user with pSIC under Rician fading channel}. The proof is completed.
\end{proof}
\end{corollary}

\subsection{The Outage Probability of User $m$}
According to the NOMA protocol, if the distant user $m$ cannot detect the refracting signal $x_m$, the outage will happen. The corresponding outage probability can be written as
\begin{align}\label{the OP of user m}
{P_m} = {\rm{Pr}}\left( {{\gamma _m} < {\gamma _{t{h_m}}}} \right).
\end{align}

\begin{theorem}\label{Theorem2:the OP of user m under Rician fading channel}
Under Rician fading channels, the approximate expression for outage probability of user $m$ for STAR-RIS-NOMA networks is given by
\begin{align}\label{the OP of user m under Rician fading channel}
{P_m} \approx {\left[ {\Gamma \left( {\frac{{K\mu _m^2}}{{{\Omega _m}}}} \right)} \right]^{ - 1}}\gamma \left( {\frac{{K\mu _m^2}}{{{\Omega _m}}},\frac{{{\mu _m}\sqrt \tau  }}{{{\Omega _m}}}} \right),
\end{align}
where  ${\mu _m} = \frac{{\pi \sqrt {{\alpha _{sr}}{\alpha _{rm}}} }}{{{4\left( {\kappa  + 1} \right)}}}{\left[ {{L_{\frac{1}{2}}}\left( { - \kappa } \right)} \right]^2}$, $\tau  = \frac{{{\gamma _{t{h_m}}}}}{{\rho \left( {{a_m} - {\gamma _{t{h_m}}}{a_n}} \right)}}$,
${\Omega _m} = {\alpha _{sr}}{\alpha _{rm}}\left\{ {{\rm{1}} - \frac{{{\pi ^{\rm{2}}}}}{{{\rm{16}}{{\left( {\kappa  + 1} \right)}^2}}}{{\left[ {{L_{\frac{1}{2}}}\left( \kappa  \right)} \right]}^{\rm{4}}}} \right\}$ and $\gamma \left( {a,x} \right) = \int_0^x {{t^{a - 1}}{e^{ - t}}dt} $ is the lower incomplete Gamma function \cite[Eq. (8.350.1)]{2000gradshteyn}. Note that \eqref{the OP of user m under Rician fading channel} is derived under condition ${a_m} > {\gamma _{t{h_m}}}{a_n}$.
\end{theorem}
\begin{proof}
See Appendix~B.
\end{proof}

\begin{proposition}\label{The system outage of STAR_RIS_NOMA}
Based on above analyses, the system outage probability of STAR-RIS-NOMA with ipSIC/pSIC over Rician fading channels is given by
\begin{align}\label{The expression for system outage of STAR_RIS_NOMA}
P_{NOMA,\Lambda }^{STAR - RIS} = 1 - \left( {1 - {P_{n,\Lambda }}} \right)\left( {1 - {P_m}} \right),
\end{align}
where  $\Lambda   \in \left( {ipSIC,pSIC} \right)$. ${P_{n,ipSIC}}$, ${P_{n,pSIC}}$ and ${P_{m}}$  can be obtained from \eqref{the OP of user n with ipSIC under Rician fading channel}, \eqref{the OP of the n-th user with pSIC under Rician fading channel} and \eqref{the OP of user m under Rician fading channel}, respectively.
\end{proposition}

\subsection{The Outage Probability of STAR-RIS-OMA}
 For STAR-RIS-OMA networks, the entire communication process includes two time slots. In the first time slot, the BS sends the information $x_n$ through RIS to reflect to user $n$, and the BS sends $x_m$ to transmit to user $m$ via the assistance of RIS in the second slot. At this moment, an outage is defined as the probability that the instantaneous SNR i.e., $\gamma _\varphi ^{OMA}$ falls bellow a threshold SNR.
 Hence the outage probability of user $n$ and user $m$ for STAR-RIS-OMA can be expressed as
\begin{align}\label{Outage Probability of User n for STAR-RIS-OMA}
P_\varphi ^{OMA} = {\text{Pr}}\left( {\gamma _\varphi ^{OMA} < \gamma _{t{h_\varphi }}^{OMA}} \right),
\end{align}
where $\gamma _{t{h_\varphi }}^{OMA} = {2^{2{R_\varphi }}} - 1$ denote the target SNRs of user $\varphi$ with detecting the signal ${x_\varphi}$.
Similar to the above derived processes, the outage probabilities of user $n$ and user $m$ for STAR-RIS-OMA networks are present in the following theorem.

\begin{theorem}\label{Theorem3:the OP of user n STAR-RIS-OMA}
Under Rician fading channels, the approximate expressions of the outage probability of user $n$ and user $m$ for STAR-RIS-OMA networks can be respectively given by
\begin{align}\label{The expression for Outage Probability of User n for STAR-RIS-OMA}
&P_n^{OMA} \approx {\sum\limits_{u = 1}^U {\frac{{\ell \left( {\kappa  + 1} \right){b_u}\left( {{x_u}{\text{ + }}1} \right)}}
{{{\alpha _{sn}}{e^\kappa }\Gamma \left( {{\varphi _n} + 1} \right)}}e} ^{ - \frac{{\ell \left( {\kappa  + 1} \right){{\left( {{x_u}{\text{ + }}1} \right)}^2}}}{{4{\alpha _{sn}}}}}}\nonumber\\
& \times {I_0}\left( {\frac{{\left( {{x_u}{\text{ + }}1} \right)\sqrt {\ell \kappa \left( {\kappa  + 1} \right)} }}
{{\sqrt {{\alpha _{sn}}} }}} \right)\gamma \left( {{\varphi _n} + 1,\frac{{\left( {1 - {x_u}} \right)\sqrt \ell  }}
{{2{\phi _n}}}} \right),
\end{align}
and
\begin{align}\label{The expression for Outage Probability of User m for STAR-RIS-OMA}
P_m^{OMA} \approx {\left[ {\Gamma \left( {\frac{{K\mu _m^2}}{{{\Omega _m}}}} \right)} \right]^{ - 1}}\gamma \left( {\frac{{K\mu _m^2}}{{{\Omega _m}}},\frac{{{\mu _m}}}{{{\Omega _m}}}\sqrt {\frac{{\gamma _{t{h_m}}^{OMA}}}{\rho {a_m}}} } \right),
\end{align}
where $\ell  = \frac{{\gamma _{t{h_n}}^{OMA}}}{\rho {a_n}}$.
\end{theorem}

\begin{proposition}\label{The system outage of STAR_RIS_OMA}
Similar to \eqref{The expression for system outage of STAR_RIS_NOMA}, the system outage probability of STAR-RIS-OMA over Rician fading channels is given by
\begin{align}\label{The expression for system outage of STAR_RIS_OMA}
P_{OMA }^{STAR - RIS} = 1 - \left( {1 - P_n^{OMA}} \right)\left( {1 - P_m^{OMA} } \right),
\end{align}
where $P_n^{OMA}$ and $P_m^{OMA}$ can be obtained from \eqref{The expression for Outage Probability of User n for STAR-RIS-OMA} and \eqref{The expression for Outage Probability of User m for STAR-RIS-OMA}, respectively.
\end{proposition}

\subsection{Diversity Analysis}\label{Diversity Analysis}

To gain better insights, the diversity order can be chosen to characterize the outage behaviors for wireless communication networks, which has ability to depict how fast the outage probability decreases with increasing SNR \cite{Alouini2005Digital,laneman2004cooperative}. In other words, the lager diversity order implies the faster decay in outage probability and more robustness to fading. To be precise, the diversity order can be expressed as
\begin{align}\label{The definition of diversity order for IRS-NOMA}
d =  - \mathop {\lim }\limits_{\rho  \to \infty } \frac{{\log \left( {P ^{\infty} \left( \rho  \right)} \right)}}{{\log \rho }},
\end{align}
where ${P ^{\infty}  \left( \rho  \right)}$ denotes the asymptotic outage probability in the high SNR regime.

Then we first provide the approximate outage probability of user $n$ with ipSIC. As can be seen that
the variable $\beta$ in \eqref{the OP of user n with ipSIC under Rician fading channel} is equal to zero at high SNRs. The corresponding outage probability of user $n$ with ipSIC is a constant, which can be provided in the following corollary.
\begin{corollary}\label{Corollary1:the asymptotic OP of user n with ipSIC under Rician fading channel}
When $\rho $ tends to $\infty$, the asymptotic expression for outage probability of user $n$ with ipSIC for STAR-RIS-NOMA networks is given by
\begin{align}\label{the asymptotic OP of user n with ipSIC under Rician fading channel}
P_{n,ipSIC}^\infty  = \Phi \sum\limits_{p = 1}^P {\sum\limits_{u = 1}^U {{H_p}{b_u}{{\left( {\tilde \chi } \right)}^{{\varphi _n}  + 1}}{{\left( {{x_u}{\rm{ + }}1} \right)}^{\varphi _n} }{e^{ - \frac{{\left( {{x_u}{\rm{ + }}1} \right)\tilde \chi }}{{2{\phi _n} }}}}} }    \nonumber     \\
  \times \left\{ {1 - Q\left( {\sqrt {2\kappa } ,\left[ {\tilde \chi  - \frac{{\left( {{x_u}{\rm{ + }}1} \right)\tilde \chi }}{2}} \right]\sqrt {\frac{{2\left( {\kappa  + 1} \right)}}{{{\alpha _{sn}}}}} } \right)} \right\} ,
\end{align}
where $\tilde \chi  = \sqrt { {x_p}{\Omega _I}\tilde \beta } $, and $\tilde \beta  = \frac{{{\gamma _{t{h_n}}}}}{{{a_n}}}$.
\end{corollary}
\begin{remark}\label{Remark1:the diversity order of user n with ipSIC under Rician fading channel}
Upon substituting \eqref{the asymptotic OP of user n with ipSIC under Rician fading channel} into \eqref{The definition of diversity order for IRS-NOMA}, a $zero$ diversity order is achieved by user $n$ with ipSIC, which is consistent with conventional cooperative NOMA communications. This is due to the impact of residual interference on its outage behaviors.
\end{remark}

To obtain the accurate diversity orders, the asymptotic outage probability of user $n$ with pSIC is calculated by exploiting the Laplace transform and convolution theorem in the following part.
\begin{corollary}\label{Corollary3:the asymptotic OP of user n with pSIC under Rician fading channel}
When $\rho $ tends to $\infty$, the asymptotic expression for outage probability of user $n$ with pSIC for STAR-RIS-NOMA networks is given by
\begin{align}\label{The asymptotic OP of user n with pSIC under Rician fading channel}
P_{n,pSIC}^{asym} = \frac{{2{\vartheta ^K}(\kappa  + 1){\beta ^{K + 1}}}}
{{{\alpha _{sn}}{{\left( {{\alpha _{sr}}{\alpha _{rn}}} \right)}^K}{e^\kappa }\left( {2K + 2} \right)!}},
\end{align}
where $\vartheta  = {}_2{F_1}\left( {2,\frac{1}
{2};\frac{5}
{2};1} \right)\frac{{16{{\left( {1 + \kappa } \right)}^2}}}
{{3\exp \left( {2\kappa } \right)}}$ and ${}_2{F_1}\left( { \cdot , \cdot ; \cdot ; \cdot } \right)$ is the ordinary hypergeometric function \cite[Eq. (9.100)]{2000gradshteyn}.
\end{corollary}
\begin{proof}
See Appendix~C.
\end{proof}
\begin{remark}\label{Remark2:the diversity order of user n with pSIC under Rician fading channel}
Upon substituting  \eqref{The asymptotic OP of user n with pSIC under Rician fading channel} into \eqref{The definition of diversity order for IRS-NOMA}, the diversity order of user $n$ with pSIC is equal to $K + 1$, which is in connection with the number of  configurable elements $K$ and direct communication link.
\end{remark}
\begin{corollary}\label{Corollary3:the asymptotic OP of user m under Rician fading channel}
Similar to the solving processes of \eqref{The asymptotic OP of user n with pSIC under Rician fading channel}, when $\rho $ tends to $\infty$, the asymptotic expression for outage probability of user $m$ for STAR-RIS-NOMA networks is given by
\begin{align}\label{the asymptotic OP of user m under Rician fading channel}
P_m^{asym} = \frac{{{\vartheta ^K}{\tau ^K}}}
{{2K{{\left( {{\alpha _{sr}}{\alpha _{rm}}} \right)}^K}\left( {2K - 1} \right)!}}.
\end{align}
\end{corollary}
\begin{remark}\label{Remark3:the diversity order of user m under Rician fading channel}
Upon substituting \eqref{the asymptotic OP of user m under Rician fading channel} into \eqref{The definition of diversity order for IRS-NOMA}, the diversity order of $K$ is achieved by user $m$ carefully, which is only related to the configurable elements.
\end{remark}
\begin{remark}\label{Remark4:the diversity order of user m OMA}
On the basis of the procedures in \eqref{The asymptotic OP of user n with pSIC under Rician fading channel} and \eqref{the asymptotic OP of user m under Rician fading channel}, the diversity orders for user $n$ and user $m$ of STAR-RIS-OMA networks are equal to $K + 1$ and $K$, respectively.
\end{remark}

\subsection{Delay-limited Transmission}\label{delay-limited mode System throughput}
In delay-limited transmission scenario, the system throughput is determined by evaluating outage probability at a constant source transmission rate i.e., $R_n$ and $R_m$ \cite{Zhong2014,Yue2020IRSNOMA}. Hence the delay-limited system throughput of STAR-RIS-NOMA with ipSIC/pSIC over Rician fading channels can be given by
\begin{align}\label{The system throughput of delay-limited mode}
{R_{dl,\Lambda  }} = \left( {1 - {P_{n,\Lambda }}} \right){R_n} + \left( {1 - {P_m}} \right){R_m} ,
\end{align}
where ${P_{n,ipSIC}}$, ${P_{n,pSIC}}$ and ${{P_m}}$ can be obtain from \eqref{the OP of user n with ipSIC under Rician fading channel}, \eqref{the OP of the n-th user with pSIC under Rician fading channel} and \eqref{the OP of user m under Rician fading channel}, respectively.

\section{Ergodic Rate}\label{Ergodic Rate}
In this section,  the ergodic performance of user $n$ and user $m$ is characterized for STAR-RIS-NOMA networks, in which the ipSIC and pSIC schemes are also taken into account. Assuming that user $n$ can successfully detect user $m$'s information $x_m$ by invoking SIC scheme and then the ergodic rate of user $n$ with ipSIC is expressed as
\begin{align}\label{the ergodic rate of user n with ipSIC under Rician fading channel}
 R_{n,ipSIC}^{erg} = \mathbb{E} \left[ {\log \left( {1 + { \frac{{{{\left| {{h_{sn}} + {\bf{h}}_{rn}^H{{\bf{\Theta}} _R}{{\bf{h}}_{sr}}} \right|}^2}\rho {a_n}}}{{\varpi {{\left| {{h_I}} \right|}^2}\rho  + 1}}}} \right)} \right].
\end{align}
As can be seen that it is difficult to obtain the exact expression solution from the above equation. However, we can evaluate it numerically by using simulation software, i.e., Matlab or Mathematica. However, the approximate ergodic rate expression of user $n$ with pSIC can be provided in the following theorem.

\begin{theorem}\label{Theorem3:the ergodic rate  of user n under Rician fading channel}
For the special case by substituting $\varpi=0$ into \eqref{the ergodic rate of user n with ipSIC under Rician fading channel}, the approximate expression for ergodic rate of user $n$ with pSIC for STAR-RIS-NOMA networks is given by \eqref{the ergodic rate of the user n with ipSIC} at the top of next page.
\begin{figure*}[!t]
\normalsize
\begin{align}\label{the ergodic rate of the user n with ipSIC}
R_{n,pSIC}^{erg}\approx \frac{{\rho {a_n}}}{{\ln 2}}\int_0^\infty  {\frac{{1 - {{\sum\limits_{u = 1}^U {\frac{{x\left( {\kappa  + 1} \right){b_u}\left( {{x_u}{\rm{ + }}1} \right)}}{{{\alpha _{sn}}{e^\kappa }\Gamma \left( {{\varphi _n} + 1} \right)}}e} }^{ - \frac{{x\left( {\kappa  + 1} \right){{\left( {{x_u}{\rm{ + }}1} \right)}^2}}}{{4{\alpha _{sn}}}}}}{I_0}\left( {\left( {{x_u}{\rm{ + }}1} \right)\sqrt {\frac{{x\kappa \left( {\kappa  + 1} \right)}}{{{\alpha _{sn}}}}} } \right)\gamma \left( {{\varphi _n} + 1,\frac{{\left( {1 - {x_u}} \right)\sqrt x }}{{2{\phi _n}}}} \right)}}{{1 + \rho {a_n}x}}} dx.
\end{align}
\hrulefill \vspace*{0pt}
\end{figure*}
\begin{proof}
Upon substituting $\varpi=0$ into \eqref{the ergodic rate of user n with ipSIC under Rician fading channel}, the ergodic rate of user $n$ with pSIC can be calculated as
\begin{align}\label{ergodic rate expression derived for user n with pSIC}
 R_{n,pSIC}^{erg} =& \mathbb{E}\left[ {\log \left( {1 + \underbrace {{{\left| {{h_{sn}} + {\bf{h}}_{rn}^H{{\bf{\Theta}} _R}{{\bf{h}}_{sr}}} \right|}^2}}_{{X_1}}\rho {a_n}} \right)} \right] \nonumber \\
  =& \frac{{\rho {a_n}}}{{\ln 2}}\int_0^\infty  {\frac{{1 - {F_{{X_1}}}\left( x \right)}}{{1 + \rho {a_n}x}}} dx.
\end{align}
By the virtue of \eqref{the OP of the n-th user with pSIC under Rician fading channel}, the CDF of $X_1$ can be approximated as
\begin{align}\label{CDF of X1}
 &{F_{{X_1}}}\left( x \right) \approx {\sum\limits_{u = 1}^U {\frac{{x\left( {\kappa  + 1} \right){b_u}\left( {{x_u}{\rm{ + }}1} \right)}}{{{\alpha _{sn}}{e^\kappa }\Gamma \left( {{\varphi _n} + 1} \right)}}e} ^{ - \frac{{x\left( {\kappa  + 1} \right){{\left( {{x_u}{\rm{ + }}1} \right)}^2}}}{{4{\alpha _{sn}}}}}} \nonumber \\
 & \times {I_0}\left( {\frac{{\left( {{x_u}{\rm{ + }}1} \right)\sqrt {x\kappa \left( {\kappa  + 1} \right)} }}{{\sqrt {{\alpha _{sn}}} }}} \right)\gamma \left( {{\varphi _n} + 1,\frac{{\left( {1 - {x_u}} \right)\sqrt x }}{{2{\phi _n}}}} \right).
\end{align}
Upon substituting \eqref{CDF of X1} into \eqref{ergodic rate expression derived for user n with pSIC}, we can obtain \eqref{the ergodic rate of the user n with ipSIC}. The proof is completed.
\end{proof}
\end{theorem}

\begin{theorem}\label{Theorem4:the ergodic rate of user m under Rician fading channel}
Under Rician fading channels, the approximate expression for ergodic rate of user $m$ for STAR-RIS-NOMA networks is approximated as
\begin{align}\label{the ergodic rate of user m under Rician fading channel}
 &{R_{m,erg}}  \approx  \frac{{\pi {a_m}}}{{U\ln 2}}\sum\limits_{n = 1}^N {\frac{{\sqrt {1 - x_u^2} }}{{2{a_n} + \left( {{x_u}{\rm{ + }}1} \right){a_m}}}}  \nonumber \\
&\times \left[ {1 - \frac{1}{{\Gamma \left( {{\varphi _m} + 1} \right)}}\gamma \left( {{\varphi _m} + 1,\sqrt {\frac{{{x_u} + 1}}{{\phi _m^2\rho {a_n}\left( {1 - {x_u}} \right)}}} } \right)} \right],
\end{align}
where ${\varphi _m} = \frac{{\mu _m^2 K}}{{{\Omega _m}}} - 1$ and ${\phi _m} = \frac{{{\Omega _m}}}{{{\mu _m}}}$.
\end{theorem}
\begin{proof}
See Appendix~D.
\end{proof}

For STAR-RIS-OMA networks,  the achievable data rate can be written as ${{\tilde R}_\varphi } = \frac{1}{2}\log \left( {1 + \gamma _\varphi ^{OMA}} \right)$. Referring to the derivation processes of \eqref{ergodic rate expression derived for user n with pSIC}, the ergodic rate of user $\varphi$ can be provided in the following corollary.
\begin{corollary}\label{corollary5:the ergodic rate of user m under Rician fading channel STAR-RIS-OMA}
Under Rician fading channels, the expressions of ergodic rate for user $n$ and user $m$ in STAR-RIS-OMA networks are respectively given by \eqref{the ergodic rate of the user n for STAR-OMA} at the top of the next page
\begin{figure*}[!t]
\normalsize
\begin{align}\label{the ergodic rate of the user n for STAR-OMA}
R_{n,OMA}^{erg} \approx \frac{\rho{a_n} }
{{2\ln 2}}\int_0^\infty  {\frac{{1 - {{\sum\limits_{u = 1}^U {\frac{{x\left( {\kappa  + 1} \right){b_u}\left( {{x_u}{\text{ + }}1} \right)}}
{{{\alpha _{sn}}{e^\kappa }\Gamma \left( {{\varphi _n} + 1} \right)}}e} }^{ - \frac{{x\left( {\kappa  + 1} \right){{\left( {{x_u}{\text{ + }}1} \right)}^2}}}
{{4{\alpha _{sn}}}}}}{I_0}\left( {\frac{{\left( {{x_u}{\text{ + }}1} \right)\sqrt {x\kappa \left( {\kappa  + 1} \right)} }}
{{\sqrt {{\alpha _{sn}}} }}} \right)\gamma \left( {{\varphi _n} + 1,\frac{{\left( {1 - {x_u}} \right)\sqrt x }}
{{2{\phi _n}}}} \right)}}
{{1 + \rho{a_n} x}}} dx.
\end{align}
\hrulefill \vspace*{0pt}
\end{figure*}
and
\begin{align}\label{the ergodic rate of the user n for STAR-RIS-OMA}
R_{m,OMA}^{erg} = \frac{\rho{a_m} }
{{2\ln 2}}\int_0^\infty  {\frac{{\Gamma \left( {{\varphi _m} + 1} \right) - \gamma \left( {{\varphi _m} + 1,\frac{{\sqrt x }}
{{{\phi _m}}}} \right)}}
{{\Gamma \left( {{\varphi _m} + 1} \right)\left( {1 + \rho{a_m} x} \right)}}} dx.
\end{align}
\end{corollary}


\subsection{Slope Analysis}\label{Slope Analysis}
Similar to the diversity order, the high SNR slope is another performance evaluating indicator aims to capture the diversification of ergodic rate with the transmitting SNRs, which can be defined as
\begin{align}\label{high SNR slope}
S = \mathop {\lim }\limits_{\rho  \to \infty } \frac{{R_{erg}^\infty \left( \rho  \right)}}{{\log \left( \rho  \right)}},
\end{align}
where ${{R_{erg}^\infty \left( \rho  \right)}}$ is the approximated ergodic rate at high SNRs.

As can be seen from \eqref{the ergodic rate of the user n with ipSIC} that it is difficult to calculate the approximate expression of ergodic rate of user $n$ with pSIC. To facilitate analyses, we try to employ the Jensen's inequality to the upper bound of ergodic rate, which can be written as
\begin{align}\label{upper bound}
 &R_{n,pSIC}^{erg} = {\mathbb{E}}\left[ {\log \left( {1 + {{\left| {{h_{sn}} + {\bf{h}}_{rn}^H{ {\bf{\Theta}} _R}{{\bf{h}}_{sr}}} \right|}^2}\rho {a_n}} \right)} \right] \nonumber \\
 & \le \log \left( {1 + \rho {a_n}{\mathbb{E}}\left[ {{{\left( {\underbrace {\left| {{h_{sn}}} \right|}_{{Y_1}} + \underbrace {\left| {{\bf{h}}_{rn}^H{ {\bf{\Theta}} _R}{{\bf{h}}_{sr}}} \right|}_{{Y_2}}} \right)}^2}} \right]} \right).
\end{align}
With the assistant of \eqref{the mean of X_k}, \eqref{the variance of X_k} and \cite[Eq. (2.16)]{2006Probability}, we can calculate the expectation and variance of $Y_1$ and $Y_2$. Hence the upper bound of ergodic rate for user $n$ with pSIC is given by
\begin{align}\label{upper bound of ergodic rate}
R_{n,pSIC}^{erg,upp} =& \log \left\langle {1 + \rho {a_n}\left\{ {\frac{{\pi {\alpha _{sn}}}}{{4\left( {1 + \kappa } \right)}}{{\left[ {{L_{\frac{1}{2}}}\left( { - \kappa } \right)} \right]}^2}} \right.} \right. + K{\Omega _n} \nonumber\\
&   + {\left( {K{\mu _n}} \right)^2} + 2K{\mu _n}\sqrt {\frac{{\pi {\alpha _{sn}}}}{{4\left( {1 + \kappa } \right)}}} {L_{\frac{1}{2}}}\left( { - \kappa } \right)\nonumber \\
& \left. {\left. { + {\alpha _{sn}}\left[ {1 - \frac{\pi }{{4\left( {1 + \kappa } \right)}}{{\left[ {{L_{\frac{1}{2}}}\left( { - \kappa } \right)} \right]}^2}} \right]} \right\}} \right\rangle  .
\end{align}
\begin{remark}\label{Remark4:the slope of user n under Rician fading channel}
Upon substituting \eqref{upper bound of ergodic rate} into \eqref{high SNR slope}, the high SNR slope of user $n$ with pSIC is equal to one, which is in line with the discussions in conventional RIS-NOMA networks. One can observe that the direct link between BS and user $n$ does not improve the high SNR slope.
\end{remark}
Based on \eqref{The SINR of the m-th user}, when $\rho $ tends to infinity, the asymptotic ergodic rate of user $m$ for STAR-RIS-NOMA networks can be given by
\begin{align}\label{the asymptotic ergodic rate of user $m$}
R_{m,erg}^\infty  = \log \left[ {1 + \left( {\frac{{{a_m}}}{{{a_n}}}} \right)} \right].
\end{align}
\begin{remark}\label{Remark5:the slope of user m under Rician fading channel}
Upon substituting \eqref{the asymptotic ergodic rate of user $m$} into \eqref{high SNR slope}, a $zero$  high SNR slope of user $m$ is obtained for STAR-RIS-NOMA networks.
\end{remark}

For the STAR-RIS-OMA, similar to the process of obtaining \eqref{upper bound of ergodic rate}, the upper bound of user $n$ and user $m$ is given by
\begin{align}\label{upper bound of user n OMA}
R_{n,OMA}^{erg,upp} =& \frac{1}{2}\log \left\langle {1 + \rho {a_n} \left\{ {\frac{{\pi {\alpha _{sn}}}}{{4\left( {1 + \kappa } \right)}}{{\left[ {{L_{\frac{1}{2}}}\left( { - \kappa } \right)} \right]}^2} + K{\Omega _n}} \right.} \right. \nonumber\\
&  + {\left( {K{\mu _n}} \right)^2} + 2K{\mu _n}\sqrt {\frac{{\pi {\alpha _{sn}}}}{{4\left( {1 + \kappa } \right)}}} {L_{\frac{1}{2}}}\left( { - \kappa } \right)\nonumber\\
&  \left. {\left. { + {\alpha _{sn}}\left[ {1 - \frac{\pi }
{{4\left( {1 + \kappa } \right)}}{{\left[ {{L_{\frac{1}
{2}}}\left( { - \kappa } \right)} \right]}^2}} \right]} \right\}} \right\rangle,
\end{align}
and
\begin{align}\label{upper bound of user m OMA}
R_{m,OMA}^{erg,upp} = \frac{1}
{2}\log \left\{ {1 + \rho{a_m} \left[ {{{\left( {K{\mu _m}} \right)}^2} + K{\Omega _m}} \right]} \right\}
\end{align}
respectively.
\begin{remark}\label{Remark5:the slope of user n and m for STAR OMA}
Upon substituting \eqref{upper bound of user n OMA} and \eqref{upper bound of user m OMA} into \eqref{high SNR slope} respectively, the high SNR slope of both user $n$ and user $m$ for STAR-RIS-OMA networks are equal to one half.
\end{remark}

\subsection{Delay-tolerant Transmission}\label{delay-tolerated mode System throughput}
In delay-tolerant transmission scenario, the BS send the information with any constant data rate bridled by the user's channel conditions \cite{Zhong2014,Yue8026173}. At this moment, the throughput of STAR-RIS-NOMA with pSIC over Rician fading channels can be given by
\begin{align}\label{The system throughput of delay-tolerant mode}
{R_{dt}} = R_{n,pSIC}^{erg} + {R_{m,erg}} ,
\end{align}
where $R_{n,pSIC}^{erg}$ and ${R_{m,erg}}$ can be obtain from \eqref{the ergodic rate of the user n with ipSIC} and \eqref{the ergodic rate of user m under Rician fading channel}, respectively.

\begin{table}[t]
\centering
\caption{Diversity order and high SNR slope for STAR-RIS-NOMA and STAR-RIS-OMA networks.}
\tabcolsep5pt
\renewcommand\arraystretch{1.85} 
\begin{tabular}{|c|c|c|c|c|}
\hline
\textbf{Mode } &  \textbf{SIC}  &  \textbf{User}  &  \textbf{D}  &  \textbf{S}  \\
\hline
\multirow{2}{*}{STAR-RIS-OMA }  & \multirow{2}{*}{\raisebox{0.5mm}{------}} &  user $n$  & $K+ 1$ &$\frac{1}{2}$ \\
\cline{3-5}
                                &  &  user $m$  & $K $  & $\frac{1}{2}$ \\
\hline
\multirow{3}{*}{STAR-RIS-NOMA } & ipSIC &  user $n$  & 0 & --- \\
\cline{2-5}
                                & pSIC &  user $n$  & $K+ 1$  & 1 \\
                                \cline{2-5}
                                & ------ &  user $m$  & $K$  & 0  \\
\hline
\end{tabular}
\label{parameter}
\end{table}
\begin{table}[!t]
\centering
\caption{The parameters for simulation results.}
\tabcolsep5pt
\renewcommand\arraystretch{1.2} 
\begin{tabular}{|l|l|}
\hline
Monte Carlo simulations repeated  &  ${10^6}$ iterations \\
\hline
\multirow{2}{*}{The power allocation factors for two users }&  \multirow{1}{*}{$a_n=0.2$}   \\
                                                        &  \multirow{1}{*}{$a_m=0.8$}   \\
\hline
\multirow{2}{*}{The targeted data rates for two users } & \multirow{1}{*}{$R_{{n}}=0.5$ BPCU}  \\
                                                & \multirow{1}{*}{$R_{{m}}=0.5$ BPCU} \\
\hline
The distance from BS to user $n$   &  $d_{sn}=10$ m   \\
\hline
\multirow{1}{*}{The distance from BS to STAR-RIS  }
                                                     & \multirow{1}{*}{$d_{sr}=8$ m } \\
\hline
\multirow{1}{*}{The distance from STAR-RIS to user $n$   }
                                                     & \multirow{1}{*}{$d_{rn}=6$ m }  \\
\hline
\multirow{1}{*}{The distance from STAR-RIS to user $m$   }
                                                     & \multirow{1}{*}{$d_{rm}=10$ m }  \\
                                                     \cline{1-2}
Pass loss expression   &  $\alpha  = 2 $   \\
\cline{1-2}

\end{tabular}
\label{parameter}
\end{table}
\section{Simulation Results}\label{Numerical Results}
In this section, we provide the simulation results to verify the theoretical analysis results derived in the above subsections for STAR-RIS-NOMA networks.
The impacts of configurable elements $K$ and Rician factor $\kappa$ on the performance of STAR-RIS-NOMA are taken into account carefully. For notational simplicity, Table~\ref{parameter} has summarized the simulation parameters used in this paper, in which BPCU is the short for bit per channel use and the fixed power allocation of non-orthogonal users considered is validity for the analytical expressions of outage probability and ergodic rate. Note that the choice of small target rates can be applied into the Internet of Thing scenarios, i.e.,  small packet service and so on.
To guarantee the accuracy of approximate expressions, the complexity-accuracy tradeoff parameters $P$ and $U$ are set to be $P = 300$ and $U=50$, respectively.  Without loss of generality, the conventional cooperative communication schemes, i.e., FD/HD DF and amplify-and-forward (AF)  relaying and STAR-RIS-OMA are selected to  be benchmarks for the purpose of comparison. It is worth pointing out that the entire communication process of STAR-RIS-OMA includes two time slots. In the first time slot, the BS sends the information ${x_n}$ through RIS to reflect to user $n$, and the BS sends ${x_m}$ to transmit to user $m$ via the assistance of RIS in the second slot. At this moment, the overall energy consumed of STAR-RIS-OMA is equal to that of STAR-RIS-NOMA from the perspective of comparison fairness.

\subsection{Outage Probability}
\begin{figure}[t!]
    \begin{center}
        \includegraphics[width=3.0in,  height=2.2in]{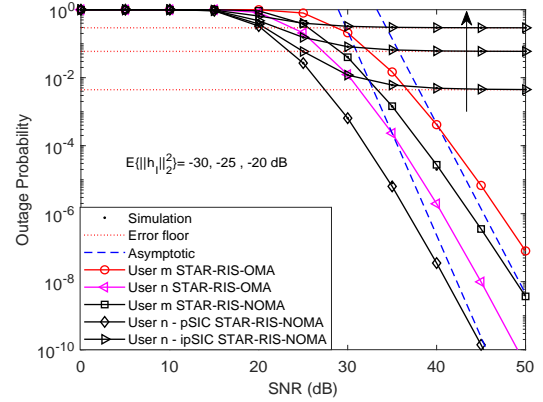}
        \caption{Outage probability versus the transmit SNR, with $K=5$, $\kappa = -5$ dB, $R_n=0.5$ and $R_m=0.5$ BPCU.}
        \label{STAR_NOMA_add_OMA}
    \end{center}
\end{figure}
Fig. \ref{STAR_NOMA_add_OMA} plots the outage probability of STAR-NOMA networks versus SNR with setting to be $K=5$, $\kappa = -5$ dB, $R_n=0.5$ and $R_m=0.5$ BPCU. The diamond and right triangle solid curves for outage probability of user $n$ with pSIC/ipSIC are plotted according to \eqref{the OP of user n with ipSIC under Rician fading channel} and \eqref{the OP of the n-th user with pSIC under Rician fading channel}, respectively. The square curve for outage probability of user $m$ is plotted based on \eqref{the OP of user m under Rician fading channel}.
The left triangle and circle solid curves for outage probability of user $n$ and user $m$ for STAR-RIS-OMA
are plotted based on \eqref{The expression for Outage Probability of User n for STAR-RIS-OMA} and \eqref{The expression for Outage Probability of User m for STAR-RIS-OMA}, respectively.
The outage probability curves are given by numerical simulation results and perfectly match with the theoretical analysis expressions derived in the above sections.
One can observe that the outage behaviors of both user $n$ with pSIC and user $m$ for STAR-RIS-NOMA are superior to that of STAR-RIS-OMA. This is due to the fact that NOMA is capable of providing better fairness compared with OMA when multiple users are served simultaneously \cite{Ding2017Mag,Yue2020IRSNOMA}.
The blue dotted curve for asymptotic outage probability of user $n$ with pSIC/ipSIC and user $m$ are plotted based on the theoretical results in \eqref{the asymptotic OP of user n with ipSIC under Rician fading channel}, \eqref{The asymptotic OP of user n with pSIC under Rician fading channel} and \eqref{the asymptotic OP of user m under Rician fading channel}, respectively. The asymptotic outage probabilities of user $n$ with ipSIC/pSIC and user $m$ match the exact performance curves in the high SNR regime, which provides an effective performance evaluation method. As can be observed that the outage behavior of user $n$ with pSIC outperforms that of user $m$ for STAR-RIS-NOMA networks. The reason is that user $n$ with pSIC can obtain the larger diversity order compared to user $m$, which is in line with the insights in \textbf{Remark \ref{Remark2:the diversity order of user n with pSIC under Rician fading channel}}.
Due to the influence of residual interference, the outage probability of user $n$ with ipSIC converges to  an error floor and thus gain a zero diversity gain.  which confirms the conclusion in \textbf{Remark \ref{Remark1:the diversity order of user n with ipSIC under Rician fading channel}}. Furthermore, with the increasing the value of residual interference, the outage performance of user $n$ with ipSIC is becoming much worse in comparison with other users. Hence it is important to take into consideration these factors in actual communication scenarios.

\begin{figure}[t!]
    \begin{center}
        \includegraphics[width=3.0in,  height=2.1in]{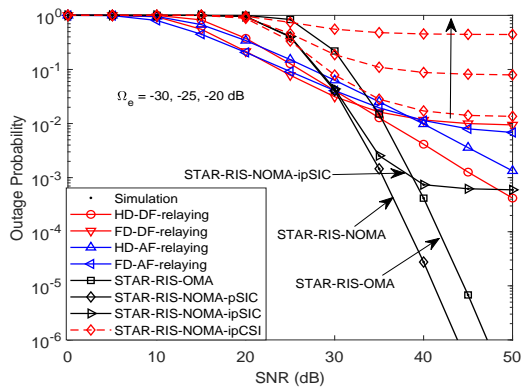}
        \caption{Outage probability versus the transmit SNR, with $K=5$, $\kappa = -5$ dB, $\mathbb{E}\{|h_{I}|^2\}=-30$ dB, $R_n=0.5$ and $R_m=0.5$ BPCU.}
        \label{STAR_add_relay_ipCSI}
    \end{center}
\end{figure}
\begin{figure}[t!]
    \begin{center}
        \includegraphics[width=3.0in,  height=2.1in]{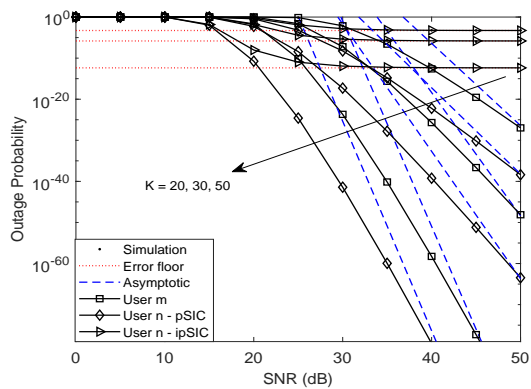}
        \caption{Outage probability versus the transmit SNR, with  $\kappa = -5$ dB, $\mathbb{E}\{|h_{I}|^2\}=-30$ dB, $R_n=2$ and $R_m=2$ BPCU.}
        \label{STAR_NOMA_diff_K}
    \end{center}
\end{figure}
To explain the superiority of STAR-RIS-NOMA, Fig. \ref{STAR_add_relay_ipCSI} plots the outage probability of STAR-NOMA networks versus SNR with different benchmarks.
It can be observed that the outage behavior of STAR-RIS-NOMA with pSIC is superior to that of STAR-RIS-OMA and conventional cooperative communication systems, i.e., HD/FD DF relays \cite{laneman2004cooperative,Kwon2010Optimal} and HD/FD AF relays \cite{laneman2004cooperative,Osorio2015}. The main reasons are that 1) The FD DF/AF relays will be affected by loop residual interference, and it needs to use the advanced cancellation technology to eliminate the interference; 2) For HD DF/AF relays, STAR-RIS-NOMA networks work in FD mode and are not affected by loop interference; and 3) The STAR-RIS-NOMA has ability to provide the higher spectrum efficiency and user fairness relative to STAR-RIS-OMA. Additionally,  the impact of channel estimation error, i.e., $\Omega_e$ on system performance are taken into consideration in  Fig. \ref{STAR_add_relay_ipCSI}. One can observe that as the increase of channel estimate errors, i.e., from $\Omega_e = -30$ dB to $\Omega_e = -20$ dB, the outage probability of STAR-RIS-NOMA with pSIC is becoming much larger and also converge to the error floors at high SNRs. As a result, it is important to consider the effect of imperfect CSI when designing practical communication systems. Furthermore, Fig. \ref{STAR_NOMA_diff_K} plots the outage probability of STAR-RIS-NOMA networks versus SNR with setting to be $\mathbb{E}\{|h_{I}|^2\}=-30$ dB $R_n=2$ and $R_m=2$ BPCU. We can be seen from the figure that as the number of  configurable elements $K$  grows, the outage probability of user $n$ and user $m$ for STAR-RIS-NOMA is getting much smaller and gain a steeper slope. This is because that the diversity orders of non-orthogonal users are related to the configurable elements at the RIS, which is also in line with the conclusions in \textbf{Remark \ref{Remark2:the diversity order of user n with pSIC under Rician fading channel}} and \textbf{Remark \ref{Remark3:the diversity order of user m under Rician fading channel}}. This phenomenon indicates that it is prerequisite to adjust the number of configurable elements involved in the work according to the different service requirements.

\begin{figure}[t!]
    \begin{center}
        \includegraphics[width=3.0in,  height=2.0in]{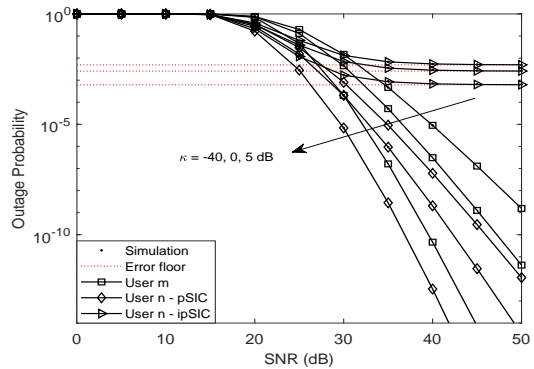}
        \caption{Outage probability versus the transmit SNR, with $K=5$, $\mathbb{E}\{|h_{I}|^2\}=-30$ dB, $R_n=0.5$ and $R_m=0.5$ BPCU.}
        \label{STAR_NOMA_diff_Kappa}
    \end{center}
\end{figure}
\begin{figure}[t!]
    \begin{center}
        \includegraphics[width=3.0in,  height=2.0in]{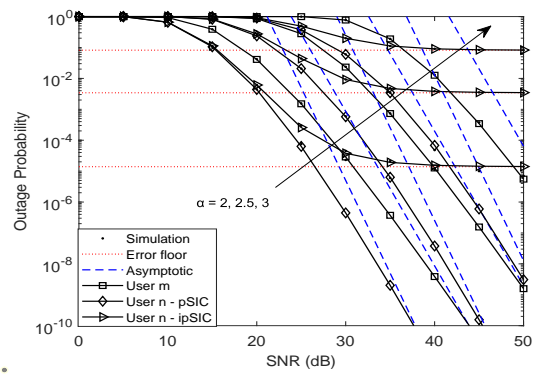}
        \caption{Outage probability versus the transmit SNR, with $K=5$, $\kappa = -5$ dB, $\mathbb{E}\{|h_{I}|^2\}=-30$ dB, $R_n=0.1$ and $R_m=0.1$ BPCU.}
        \label{STAR_NOMA_diff_alpha}
    \end{center}
\end{figure}
Fig. \ref{STAR_NOMA_diff_Kappa} plots the outage probability of STAR-RIS-NOMA networks versus SNR for the simulation with different Rician factors and $\mathbb{E}\{|h_{I}|^2\}=-30$ dB. We can observe that the Rician factor i.e., $\kappa $ has a relatively large impact on network performance of STAR-RIS-NOMA, where the outage probability of user $n$ and user $m$ decreases with the increasing of Rician factor values, i.e., $\kappa = 0$ dB to $\kappa = 5$ dB. This phenomenon can be explained that the LoS components of Rician fading channels dominate the network performance of STAR-RIS-NOMA.
Another observation is that as the Rician factor grows, i.e., $\kappa = -40$ dB to $\kappa = 0$ dB,  the outage probability of user $n$ and user $m$ for STAR-RIS-NOMA networks has minor changes. This is due to the fact that the cascade channels from the BS to RIS, and then RIS to user $n$ and user $m$ have been aligned by invoking coherent phase shifting, where the equivalent channels have a non-zero mean. Additionally, it is worth pointing out that the channels between the BS and user $n$ also include the direct link from the BS to user $n$ except the cascade channels from BS to RIS, and then RIS to user $n$, which also result in the closer outage performance under Rayleigh and Rician fading channels.

As a further development, Fig. \ref{STAR_NOMA_diff_alpha} the outage probability of STAR-RIS-NOMA networks versus SNR for the simulation with different pass loss expressions and $\mathbb{E}\{|h_{I}|^2\}=-30$ dB. One can make the following observation from figure that with the decreasing of pass loss expression, the outage behaviors of user $n$ and user $m$ are becoming much worse in different communication environment. This is due to the fact that the pass loss expression is mainly determined by the propagation environment. When $\alpha$ is relatively large, it indicates that there are many obstacles in the communication scenarios. This also confirms that the STAR-RIS can be deployed to provide the LoS transmissions.
\begin{figure}[t!]
   \begin{center}
        \includegraphics[width=3.0in,  height=2.0in]{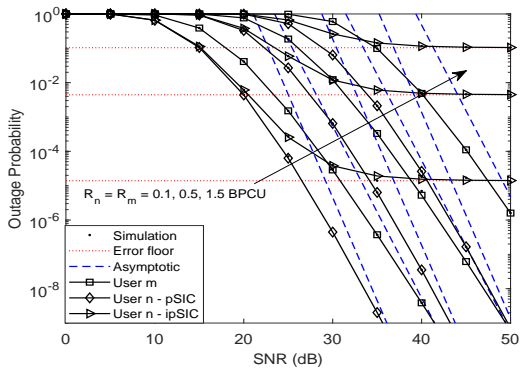}
        \caption{Outage probability versus the transmit SNR, with $K=5$, $\kappa = -5$ dB and $\mathbb{E}\{|h_{I}|^2\}=-30$ dB.}
        \label{STAR_NOMA_diff_Rate}
    \end{center}
\end{figure}
In addition, Fig. \ref{STAR_NOMA_diff_Rate} plots the outage probability of STAR-RIS-NOMA networks versus SNR with setting to be $K=5$, $\kappa = -5$ dB, $\mathbb{E}\{|h_{I}|^2\}=-30$ dB, and ${R_n} = {R_m} = 0.1,0.5,1.5$ BPCU.
It is observed that as the target rates increase, the larger outage probabilities are achieved for STAR-RIS-NOMA networks. The reason is that the achievable rates are directly combined with the target SNRs. It is favorable to decode the superposed signals for the user pairing selected with smaller target SNRs.

\begin{figure}[t!]
   \begin{center}
        \includegraphics[width=3.0in,  height=2.0in]{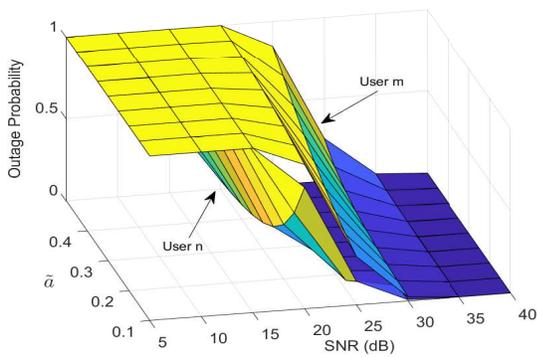}
        \caption{Outage probability versus the transmit SNR and ${\tilde a}$, with $K=5$, $\kappa = -5$ dB and $\mathbb{E}\{|h_{I}|^2\}=-30$ dB.}
        \label{STAR_RIS_NOMA_3D}
    \end{center}
\end{figure}

Fig. \ref{STAR_RIS_NOMA_3D} plots the outage probability versus SNR and the dynamic power allocation factor $\tilde a  \in  \left( {0,1} \right)$, with $K=5$, $\kappa = -5$ dB, $R_n=0.5$ and $R_m=0.5$ BPCU. Let ${a_n} = \tilde a $ and  ${a_m} = 1 -\tilde a $, which also satisfies the relationship with $\tilde a < \frac{1}{{{\gamma _{t{h_m}}} + 1}}$. The analytical curves of outage probability of user $n$ with pSIC and user $m$ are plotted according to (15) and (18), respectively. One can observe that with the value of ${\tilde a}$ increasing, the performance of user $n$ with pSIC becomes better, while the outage behavior of user $m$ deteriorates gradually. This is due to the fact that user $m$ suffers from more interference when it detects its own information.  Hence it is critical to seek out the optimal power allocation factors for balancing the performance of two users.
\subsection{Ergodic Rate}
\begin{figure}[t!]
    \begin{center}
        \includegraphics[width=3.0in,  height=2.0in]{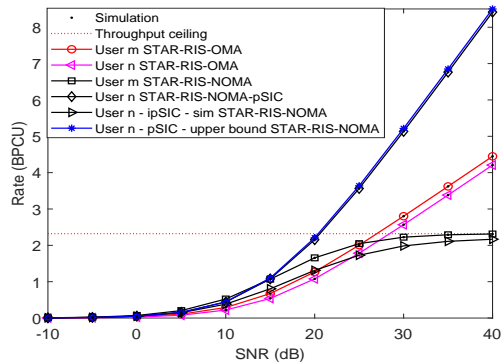}
        \caption{Rate versus the transmit SNR, with $K=20$, $\mathbb{E}\{|h_{I}|^2\}=-30$ dB, and $\kappa = -5$ dB.}
        \label{STAR_Ergodic_Rate}
    \end{center}
\end{figure}
\begin{figure}[t!]
    \begin{center}
        \includegraphics[width=3.0in,  height=2.0in]{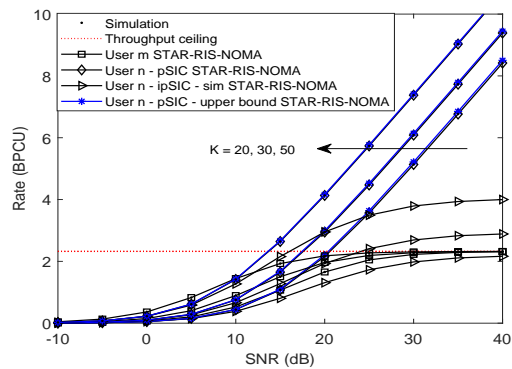}
        \caption{Rate versus the transmit SNR, with $\kappa = -5$ dB and $\mathbb{E}\{|h_{I}|^2\}=-30$ dB.}
        \label{STAR_Ergodic_Rate_diff_K}
    \end{center}
\end{figure}
Fig. \ref{STAR_Ergodic_Rate} plots the ergodic rates versus SNR, with $K=20$ and $\kappa = -5$ dB.
The right diamond and square solid curves for ergodic rates of user $n$ with pSIC and user $m$ for STAR-RIS-NOMA networks are plotted based on \eqref{the ergodic rate of the user n with ipSIC} and \eqref{the ergodic rate of user m under Rician fading channel}, respectively. Furthermore, the upper bound curve for ergodic rate of user $n$ with pSIC is plotted based on \eqref{upper bound of ergodic rate}, which can be better close to the theoretical expression.  One can observe that the ergodic rate of user $m$ converges to a throughput ceiling and thus obtain a zero high SNR slope, which is in line with the discussion in \textbf{Remark \ref{Remark5:the slope of user m under Rician fading channel}}.
The right triangle solid curve for ergodic rate of user $n$ with ipSIC is plotted according to \eqref{the ergodic rate of user n with ipSIC under Rician fading channel} by invoking Matlab simulation software. Due to the influence of residual interference, the ergodic rate of user $n$ with ipSIC tends to the constant
value at high SNRs. In addition, we can see that the ergodic rate of user $n$ with pSIC outperforms that of orthogonal users in the high SNR regime, while the ergodic rate of user $m$ is inferior to that of orthogonal user. This is due to the fact the user $n$ with pSIC gets a larger high SNR slope compared to orthogonal user. However, the high SNR slope of user $m$ is equal to zero, which is less than orthogonal user.
As a further advance, Fig. \ref{STAR_Ergodic_Rate_diff_K} plots the ergodic rates versus SNR for a simulation system with different reconfigurable elements. With increasing of reconfigurable elements, the ergodic rates of user $n$ with pSIC for STAR-RIS-NOMA networks are becoming much larger relative to that of user $m$.

\subsection{System Throughput}
\begin{figure}[t!]
    \begin{center}
        \includegraphics[width=3.0in,  height=2.0in]{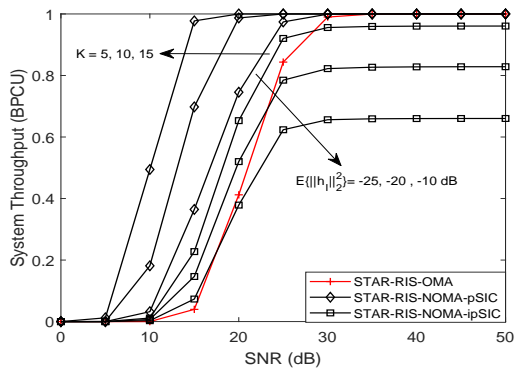}
          \caption{System throughput in delay-limited transmission mode versus SNR, with $K=5$ and $\kappa = -5$ dB.}
        \label{STAR_NOMA_delay_L}
    \end{center}
\end{figure}
\begin{figure}[t!]
    \begin{center}
        \includegraphics[width=3.0in,  height=2.0in]{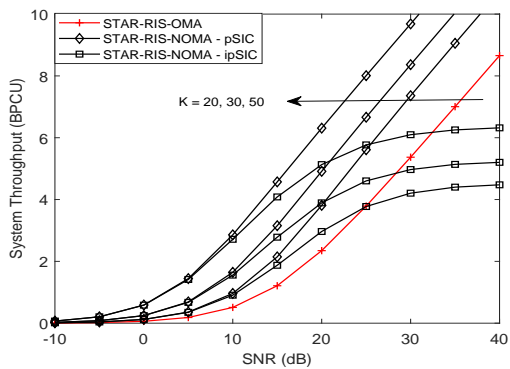}
          \caption{System throughput in delay-tolerant transmission mode versus SNR, with $\mathbb{E}\{|h_{I}|^2\}=-30$ dB and $\kappa = -5$ dB.}
        \label{STAR_delay_T}
    \end{center}
\end{figure}
Fig. \ref{STAR_NOMA_delay_L} plots the system throughput versus the transmit SNR for STAR-RIS-NOMA networks in the delay-limited transmission mode, with $K =5$, $\kappa = -5$ dB and  ${R_{n}}={R_{m}}=0.5$ BPCU. The system throughput curves of STAR-RIS-NOMA networks with ipSIC/pSIC is plotted according to \eqref{The system throughput of delay-limited mode}. We can observe from the figure that the system throughput of STAR-RIS-NOMA with pSIC are superior to that of STAR-RIS-OMA at high SNRs. This is due to the fact that the system throughput in the delay-limited transmission mode is affected by the outage probability. Due to the affect of residual interference, the system throughput of STAR-RIS-NOMA with ipSIC is worse than that of STAR-RIS-OMA. Hence it is important to consider the impact of ipSIC on STAR-RIS-NOMA network performance in practical scenarios. With the increasing of reflection elements, the STAR-RIS-NOMA networks are capable of providing the enhanced system throughput. This phenomenon can be explained as that the lower outage probability can be obtained by both the user $n$ and user $m$.
Furthermore, Fig. \ref{STAR_delay_T} plots the system throughput versus the transmit SNR for STAR-RIS-NOMA networks in the delay-tolerant transmission mode, with $K = 5$ and $\kappa = -5$ dB. The system throughput of STAR-RIS-NOMA networks with ipSIC/pSIC based on \eqref{The system throughput of delay-tolerant mode}.  One can observe that the system throughput of STAR-RIS-NOMA networks with pSIC outperforms that of STAR-RIS-NOMA with ipSIC and STAR-RIS-OMA.
As the number of reconfigurable elements increases, the SATR-RIS-NOMA networks is capable of achieving the enhanced system throughput. 

\section{Conclusion}\label{Conclusion}
In this paper, the STAR-RIS assisted downlink NOMA communication scenarios have been discussed in detail. More especially, we have investigated outage probability and ergodic rate of STAR-RIS-NOMA networks over Rician fading channels. The approximate expressions for outage probability of user $n$ and user $m$. Based on the asymptotic results, the diversity orders of user $n$ and user $m$ are obtained. It has been shown that the outage probability of STAR-RIS-NOMA outperforms that of STAR-RIS-OMA and conventional cooperative communication systems. Furthermore, the theoretical expressions of ergodic rate for user $n$ with pSIC and user $m$ have also been derived and the corresponding high SNR slopes were provided carefully. Numerical results indicated that the ergodic rate of user $n$ with pSIC outperforms that of orthogonal users at high SNRs. Moreover, the system throughput of STAR-RIS-NOMA has been surveyed in delay-limited and delay-tolerant modes.  From the perspective of practical applicability, STAR-RIS-NOMA is capable of satisfying more stringent quality-of-service requirements, where the user $n$ and user $m$ can be a high-date rate video streaming user and a low-date rate user, respectively. The setup of perfect CSI my bring about the overestimated performance for STAR-RIS-NOMA networks, hence our future work will consider the impact of imperfect CSI and seek efficient channel estimation methods. Another promising future research direction is to design the deployment locations of STAR-RIS to balance the number of users in reflecting and transmitting half-spaces.


\appendices
\section*{Appendix~A: Proof of Theorem \ref{Theorem1:the OP of user n with ipSIC under Rician fading channel}} \label{Appendix:A}
\renewcommand{\theequation}{A.\arabic{equation}}
\setcounter{equation}{0}

\begin{figure*}[!t]
\normalsize
\begin{align}\label{the total OP expression of user n with ipSIC}
{P_{n,ipSIC}} =& {\rm{Pr}}\left( {\frac{{{{\left| {{h_{sn}} + {\bf{h}}_{rn}^H{{\bf{\Theta}} _R}{{\bf{h}}_{sr}}} \right|}^2}{a_m}\rho }}{{{{\left| {{h_{sn}} + {\bf{h}}_{sr}^H{{\bf{\Theta}} _R}{{\bf{h}}_{rn}}} \right|}^2}{a_n}\rho  + 1}} < {\gamma _{t{h_m}}}} \right)    \nonumber \\
&+ {\rm{Pr}}\left( {\frac{{{{\left| {{h_{sn}} +{\bf{h}}_{rn}^H{{\bf{\Theta}} _R}{{\bf{h}}_{sr}}} \right|}^2}{a_n}\rho }}{{\varpi {{\left| {{h_I}} \right|}^2}\rho  + 1}} < {\gamma _{t{h_n}}},\frac{{{{\left| {{h_{sn}} + {\bf{h}}_{rn}^H{{\bf{\Theta}} _R}{{\bf{h}}_{sr}}} \right|}^2}{a_m}\rho }}{{{{\left| {{h_{sn}} + {\bf{h}}_{sr}^H{{\bf{\Theta}} _R}{{\bf{h}}_{rn}}} \right|}^2}{a_n}\rho  + 1}} > {\gamma _{t{h_m}}}} \right).
\end{align}
\hrulefill \vspace*{0pt}
\end{figure*}
The proof process starts by substituting \eqref{The SINR of the n-th user to detect the m-th user} and \eqref{The SINR of the n-th user} into \eqref{the OP of user n with ipSIC}, the outage probability of user $n$ with ipSIC can be further expressed as \eqref{the total OP expression of user n with ipSIC} at the top of next page. By invoking the coherent phase shifting and some arithmetic transformations, the expression of outage probability of user $n$ with ipSIC can be calculated as
\begin{align}\label{the further OP expression of user n with ipSIC}
{P_{n,ipSIC}} =\Pr \left[{{{\left| {{h_{sn}} + {\bf{h}}_{rn}^H{\Theta _R}{{\bf{h}}_{sr}}} \right|}^2} < \beta \left( {\varpi {{\left| {{h_I}} \right|}^2}\rho  + 1} \right)}\right]\nonumber \\
 = \Pr \left[ {\underbrace {{{\left| {\left| {{h_{sn}}} \right| + \sum\nolimits_{k = 1}^K {\left| {h_{rn}^kh_{sr}^k} \right|} } \right|}^2}}_Z < \beta \left( {\varpi {{\left| {{h_I}} \right|}^2}\rho  + 1} \right)} \right] ,
\end{align}
where $\varpi  = 1$ and $\beta  = \frac{{{\gamma _{t{h_n}}}}}{{{a_n}\rho }}$. The next emphasis is to solve the CDF of variable $Z$ on the left side of the above inequality.

Let $X = \sum\nolimits_{k = 1}^K {\left| {h_{rn}^kh_{sr}^k} \right|} $ and observe that it is difficult to calculate the PDF or CDF of $X$ from \eqref{The CDF of cascade Rician channels}. Since the cascade Rician fading channels has two characteristic features i.e.,  1) the PDF has a single maximum; and 2) The PDF has tails extending to infinity on both sides of the maximum. Hence we can apply the series of Laguerre polynomials to approximate this type of PDF. With the help of \cite[Eq. (2.76)]{Primak2004}, the PDF of $X$ can be approximated as
\begin{align}\label{the PDF of X approximated }
{f_X}\left( x \right) \approx \frac{{{x^{\varphi_{n}} }}}{{{{\phi_{n} }^{\varphi_{n}  + 1}}\Gamma \left( {\varphi_{n}  + 1} \right)}}\exp \left( { - \frac{x}{\phi_{n} }} \right),
\end{align}
where $\varphi _{n}= \frac{{{{\mu _n^2}}K}}{{\Omega _n}} - 1$ and $\phi_{n}  = \frac{{\Omega _n}}{{{\mu _n}}}$. As a further development, combining \eqref{The CDF of Rician channels} and \eqref{the PDF of X approximated }, and then applying some manipulates, the PDF of $Z$ can be given by
\begin{align}\label{the PDF of Z}
 {F_Z}\left( z \right) =& \int_0^{\sqrt z } {\frac{{{x^{\varphi_n} }{e^{ - \frac{x}{\phi_n }}}}}{{{b^{\varphi_n  + 1}}\Gamma \left( {\varphi_n  + 1} \right)}}\left[ 1 \right.}     \nonumber \\
 & \left. { - Q\left( {\sqrt {2\kappa } ,\left( {\sqrt z  - x} \right)\sqrt {\frac{{2\left( {\kappa  + 1} \right)}}{{{\alpha _{sn}}}}} } \right)} \right] dx.
\end{align}
Upon substituting ${f_{{{\left| {{h_I}} \right|}^2}}}\left( y \right) = \frac{1}{{{\Omega _{I}}}}{e^{ - \frac{y}{{{\Omega _{I}}}}}}$ and \eqref{the PDF of Z} into \eqref{the further OP expression of user n with ipSIC}, the outage probability of user $n$ with ipSIC can be expressed as
\begin{align}\label{Appendix:the OP of user n with ipSIC under Rician fading channel}
 &  {P_{n,ipSIC}} \approx \int_0^\infty  {\int_0^{\sqrt {\xi \left( {\varpi y\rho  + 1} \right)} } {\frac{{{x^{\varphi_n} }{e^{ - \frac{y}{{{\Omega _{I}}}} - \frac{x}{\phi_n }}}}}{{{\phi _n^{{\varphi _n} + 1}}\Gamma \left( {\varphi_n  + 1} \right){\Omega _{I}}}}\left\{ 1 \right.} } \nonumber \\
 & - \left. {Q\left( {\sqrt {2\kappa } ,\left[ {\sqrt {\xi \left( {\varpi y\rho  + 1} \right)}  - x} \right]\sqrt {\frac{{2\left( {\kappa  + 1} \right)}}{{{\alpha _{sn}}}}} } \right)} \right\}dxdy .
\end{align}
To calculate the definite integral in the above formula,  the Gauss-Chebyshev quadrature is employed to approximate this type of integral \cite[Eq. (8.8.4)]{Hildebrand1987introduction}. With the help of \cite[Eq. (8.8.4)]{Hildebrand1987introduction}, the definite integral can be approximated as
\begin{align}\label{Gauss-Chebyshev:the OP of user n with ipSIC under Rician fading channel}
 & {P_{n,ipSIC}} \approx \Phi \sum\limits_{u = 1}^U {{b_u}{{\left( {{x_u}{\rm{ + }}1} \right)}^{\varphi_n }}} \int_0^\infty  {{{\left( {\tilde \chi } \right)}^{{\varphi_n }  + 1}}{e^{ - \frac{{\left( {{x_u}{\rm{ + }}1} \right)\tilde \chi }}{{2{\phi_n} }}}}} \nonumber  \\
  & \times \left\{ {1 - Q\left( {\sqrt {2\kappa } ,\left[ {\tilde \chi  - \frac{{\left( {{x_u}{\rm{ + }}1} \right)\tilde \chi }}{2}} \right]\sqrt {\frac{{2\left( {\kappa  + 1} \right)}}{{{\alpha _{sn}}}}} } \right)} \right\}{e^{ - x}}dx ,
\end{align}
where ${b_u} = \frac{\pi }{{2U}}\sqrt {1 - x_u^2} $, $\Phi  = \frac{1}{{{2^{{\varphi _n}}}\phi _n^{{\varphi _n} + 1}\Gamma \left( {{\varphi _n} + 1} \right)}}$ and $\tilde \chi  = \sqrt {\xi \left( {\varpi x{\Omega _I}\rho  + 1} \right)} $. As a further advance, we use the Gauss-Laguerre quadrature to calculate the above indefinite integral \cite[Eq. (8.6.5)]{Hildebrand1987introduction}
 After some algebraic manipulations, we can obtain \eqref{the OP of user n with ipSIC under Rician fading channel}. The proof is completed.
\section*{Appendix~B: Proof of Theorem \ref{Theorem2:the OP of user m under Rician fading channel}} \label{Appendix:B}
\renewcommand{\theequation}{B.\arabic{equation}}
\setcounter{equation}{0}
Upon substituting \eqref{The SINR of the m-th user} into \eqref{the OP of user m}, the outage probability of user $m$ can be calculated as
\begin{align}\label{the expression of OP for m}
{P_m} = {\rm{Pr}}\left( {{{\left| {{\bf{h}}_{rm}^H{\Theta _T}{{\bf{h}}_{sr}}} \right|}^2} < {\tau }} \right).
\end{align}
We also use the coherent phase shifting scheme to deal with the correlated Rician fading channels. Similar to the solving processes of \eqref{the PDF of X approximated }, the outage probability of user $m$ can be further calculated as
\begin{align}\label{the expression of OP for m further expression}
{P_m} = & {\rm{Pr}}\left( {\left| {\sum\limits_{k = 1}^K {h_{rm}^kh_{sr}^k} } \right| < \sqrt \tau  } \right) \nonumber  \\
 =&  \frac{1}{{{{\left( {{\phi _m}} \right)}^{{\varphi _m}{\rm{ + }}1}}\Gamma \left( {{\varphi _m} + 1} \right)}}\int_0^{\sqrt \tau  } {{y^{{\varphi _m}}}{e^{ - \frac{y}{{{\phi _m}}}}}dy},
\end{align}
where ${\Omega _m} = {\alpha _{sr}}{\alpha _{rm}}\left\{ {{\rm{1}} - \frac{{{\pi ^{\rm{2}}}}}{{{\rm{16}}{{\left( {\kappa  + 1} \right)}^2}}}{{\left[ {{L_{\frac{1}{2}}}\left( \kappa  \right)} \right]}^{\rm{4}}}} \right\}$, ${\mu _m} = \frac{{\pi \sqrt {{\alpha _{sr}}{\alpha _{rm}}} }}{{4\kappa  + 1}}{\left[ {{L_{\frac{1}{2}}}\left( { - \kappa } \right)} \right]^2}$,
 ${\varphi _m} = \frac{{{{\left( {K{\mu _m}} \right)}^2}}}{{K{\Omega _m}}} - 1$, ${\phi _m} = \frac{{{\Omega _m}}}{{{\mu _m}}}$ and $\tau  = \frac{{{\gamma _{t{h_m}}}}}
{{\rho \left( {{a_m} - {\gamma _{t{h_m}}}{a_n}} \right)}}$.

By the virtue of \cite[Eq. (8.350.1)]{2000gradshteyn} and applying some arithmetic operations, we can obtain \eqref{the OP of user m under Rician fading channel}. The proof is completed.

\section*{Appendix~C: Proof of Corollary \ref{Corollary3:the asymptotic OP of user n with pSIC under Rician fading channel}} \label{Appendix:C}
\renewcommand{\theequation}{C.\arabic{equation}}
\setcounter{equation}{0}

Upon substituting $\varpi  = 0$ into \eqref{the further OP expression of user n with ipSIC}, the outage probability of user $n$ with pSIC can be expressed as
\begin{align}\label{AppendixC the further OP expression of user n with pSIC}
{P_{n,pSIC}} = {\text{Pr}}\left[ {\underbrace {\left| {\left| {{h_{sn}}} \right| + \sum\nolimits_{k = 1}^K {\left| {h_{sr}^kh_{rn}^k} \right|} } \right|}_V < \sqrt \beta  } \right].
\end{align}
To calculate the asymptotic outage probability, the following emphases are to solve the approximated PDF and CDF of variable $V$ at high SNRs.

Let ${X_k} = \left| {h_{sr}^kh_{rn}^k} \right|$ and applying \cite[Eq. (2.16.6.3)]{Prudnikov1986} into \eqref{The CDF of cascade Rician channels}, the Laplace transform expression of the PDF for $X_k$ can be given by
\begin{align}\label{AppendixC Laplace transform expression $X_k$}
&\mathcal{L}\left[ {{f_{{X_k}}}\left( x \right)} \right]\left( s \right) =  \sum\limits_{i = 0}^\infty  {\sum\limits_{j = 0}^\infty  {\frac{{{4^{i - j + 1}}\sqrt \pi  {\kappa ^{i + j}}{{\left( {1 + \kappa } \right)}^{2\left( {i + 1} \right)}}}}
{{{\alpha _{sr}}{\alpha _{rn}}{{\left( {i!} \right)}^2}{{\left( {j!} \right)}^2}{e^{2\kappa }}}}} }  \nonumber \\
&\times  {}_2{F_1}\left( {2i + 2,i - j + \frac{1}
{2};i + j + \frac{5}
{2};\frac{{s - 2\left( {\kappa  + 1} \right)}}
{{s + 2\left( {\kappa  + 1} \right)}}} \right)\nonumber \\
&\times  \frac{{\Gamma \left( {2i + 2} \right)\Gamma \left( {2j + 2} \right)}}
{{\Gamma \left( {i + j + \frac{5}
{2}} \right){{\left[ {s + 2\left( {\kappa  + 1} \right)} \right]}^{2i + 2}}}} ,
\end{align}
where ${}_2{F_1}\left( { \cdot , \cdot ; \cdot ; \cdot } \right)$ is the ordinary hypergeometric function \cite[Eq. (9.100)]{2000gradshteyn}.
When $s \to \infty $, i.e.,   $x \to 0$ the expression ${\frac{{s - 2\left( {\kappa  + 1} \right)}} {{s + 2\left( {\kappa  + 1} \right)}}}$ is approximate to 1 and the expression $\left[ {s + 2\left( {\kappa  + 1} \right)} \right]$ is dominated by $s$. Hence the Laplace transform expression is finally derived as
\begin{align}\label{AppendixC Laplace transform expression $X_k$  1}
&\mathcal{L}\left[ {{f_{{X_k}}}\left( x \right)} \right]\left( s \right) = \sum\limits_{i = 0}^\infty  {\sum\limits_{j = 0}^\infty  {\frac{{{4^{i - j + 1}}\sqrt \pi  {\kappa ^{i + j}}{{\left( {1 + \kappa } \right)}^{2\left( {i + 1} \right)}}}}
{{{\alpha _{sr}}{\alpha _{rn}}{{\left( {i!} \right)}^2}{{\left( {j!} \right)}^2}{e^{2\kappa }}}}} } \nonumber \\
 & \times  {}_2{F_1}\left( {2i + 2,i - j + \frac{1}
{2};i + j + \frac{5}
{2};1} \right)\frac{{\Gamma \left( {2i + 2} \right)\Gamma \left( {2j + 2} \right)}}
{{\Gamma \left( {i + j + \frac{5}
{2}} \right){s^{2i + 2}}}}.
\end{align}
As a further advance, the Laplace transform expression of the PDF of $X$ by invoking the convolution theorem can be given by
\begin{align}\label{AppendixC Laplace transform expression $X$  2}
\mathcal{L}\left[ {{f_X}\left( x \right)} \right]\left( s \right)  = {\left[ {\sum\limits_{i = 0}^\infty  {\sum\limits_{j = 0}^\infty  {\frac{{\theta \left( {i,j} \right)}}
{{{\alpha _{sr}}{\alpha _{rn}}}}} } {s^{ - 2i - 2}}} \right]^K},
\end{align}
where $\theta \left( {i,j} \right) = \sum\limits_{i = 0}^\infty  {\sum\limits_{j = 0}^\infty  {\frac{{{4^{i - j + 1}}\sqrt \pi  {\kappa ^{i + j}}{{\left( {1 + \kappa } \right)}^{2\left( {i + 1} \right)}}\Gamma \left( {2i + 2} \right)}}{{{{\left( {i!} \right)}^2}{{\left( {j!} \right)}^2}{e^{2\kappa }}\Gamma {{\left( {2j + 2} \right)}^{ - 1}}\Gamma \left( {i + j + \frac{5}{2}} \right)}}} } {}_2{F_1}\left( {2i } \right. \nonumber \\
 \left. {+ 2,i - j + \frac{1}{2};i + j + \frac{5}{2};1} \right) $.
We only take the first item of the two series in the above equation. The Laplace transform expression of the PDF of $X$ at high SNRs can be finally approximated  as
\begin{align}\label{AppendixC Laplace transform expression $X$  3}
\mathcal{L}\left[ {{f_X}\left( x \right)} \right]\left( s \right) \approx {\left( {\frac{\vartheta}
{{{\alpha _{sr}}{\alpha _{rn}}}}} \right)^K}{s^{ - 2K}},
\end{align}
where $\vartheta  = \theta \left( {0,0} \right) = {}_2{F_1}\left( {2,\frac{1}
{2};\frac{5}
{2};1} \right)\frac{{16{{\left( {1 + \kappa } \right)}^2}}}
{{3{e^{2\kappa }}}}$.

On the basis of \eqref{The PDF of Rician channels}, the PDF of $\left| {{h_{sn}}} \right|$ can be further expressed in the form of series as
$ {f_{\left| {{h_{sn}}} \right|}}(x) = \frac{{2x(\kappa  + 1)}}
{{{\alpha _{sn}}{e^\kappa }}} + o\left( {{x^2}} \right)$, where $o\left(  \cdot  \right)$ is the little-O notation, and $o\left( {{x^2}} \right)$ denotes a function which is asymptotically smaller than ${{x^2}}$.
Hence the Laplace transform expression of the PDF for $\left| {{h_{sn}}} \right|$ can be approximated as
\begin{align}\label{AppendixC Laplace transform expression $X$  5}
\mathcal{L}\left[ {{f_{\left| {{h_{sn}}} \right|}}\left( x \right)} \right]\left( s \right) \approx \frac{{2{s^{ - 2}}(\kappa  + 1)}}
{{{\alpha _{sn}}{e^\kappa }}}.
\end{align}

Combining \eqref{AppendixC Laplace transform expression $X$  3} and \eqref{AppendixC Laplace transform expression $X$  5}, the Laplace transform expression of the PDF for the variable $V$ at high SNRs can be given by
\begin{align}\label{AppendixC Laplace transform expression $X$  6}
\mathcal{L}\left[ {{f_V}\left( x \right)} \right]\left( s \right) \approx \frac{{2{\vartheta ^K}(\kappa  + 1)}}
{{{\alpha _{sn}}{{\left( {{\alpha _{sr}}{\alpha _{rn}}} \right)}^K}{e^\kappa }}}{s^{ - 2K - 2}}.
\end{align}
By further applying the inverse Laplace transform in the above equation, the PDF of $V$ at high SNRs can be given by
\begin{align}\label{AppendixC Laplace transform expression $X$ 7}
{f_V}\left( x \right) \approx \frac{{2{\vartheta ^K}(\kappa  + 1){x^{2K + 1}}}}
{{{\alpha _{sn}}{{\left( {{\alpha _{sr}}{\alpha _{rn}}} \right)}^K}{e^\kappa }\left( {2K + 1} \right)!}}.
\end{align}

Upon substituting \eqref{AppendixC Laplace transform expression $X$ 7} into \eqref{AppendixC the further OP expression of user n with pSIC}, we can obtain \eqref{The asymptotic OP of user n with pSIC under Rician fading channel}. The proof is completed.

\section*{Appendix~D: Proof of Theorem \ref{Theorem4:the ergodic rate of user m under Rician fading channel}} \label{Appendix:D}
\renewcommand{\theequation}{D.\arabic{equation}}
\setcounter{equation}{0}

Based on \eqref{The SINR of the m-th user}, the ergodic rate of user $m$ for STAR-RIS-NOMA networks can be expressed as
\begin{align}\label{the expression of OP for m further expression}
{R_{m,erg}} =  & \mathbb{E}\left[ {\log \left( {1 + \underbrace {\frac{{{{\left| {{\bf{h}}_{rm}^H{{\bf{\Theta}} _T}{{\bf{h}}_{sr}}} \right|}^2}\rho {a_m}}}{{{{\left| {{\bf{h}}_{rm}^H{\Theta _T}{{\bf{h}}_{sr}}} \right|}^2}\rho {a_n} + 1}}}_{{X_2}}} \right)} \right] \nonumber \\
 = & \frac{1}{{\ln 2}}\int_0^\infty  {\frac{{1 - {F_{{X_2}}}\left( x \right)}}{{1 + x}}} dx.
\end{align}
Using the coherent phase shifting scheme and with the help of \eqref{the OP of the n-th user with pSIC under Rician fading channel}, the CDF of $X_2$ is given by
\begin{align}\label{the CDF of X2}
{F_{{X_2}}}\left( x \right) = \frac{1}{{\Gamma \left( {{\varphi _m} + 1} \right)}}\gamma \left( {{\varphi _m} + 1,\frac{1}{{{\phi _m}}}\sqrt {\frac{x}{{\rho \left( {{a_m} - x{a_n}} \right)}}} } \right),
\end{align}
where ${a_m} > x{a_n}$. Upon substituting \eqref{the CDF of X2} into \eqref{the expression of OP for m further expression},  the ergodic rate of user $m$ can be approximated as
\begin{align}\label{the CDF of X2}
 &{R_{m,erg}} \approx \frac{1}{{\ln 2}}\int_0^{\frac{{{a_m}}}{{{a_n}}}} {\frac{1}{{1 + x}}} \nonumber  \\
 & - \frac{1}{{\left( {1 + x} \right)\Gamma \left( {{\varphi _m} + 1} \right)}}\gamma \left( {{\varphi _m} + 1,\frac{1}{{{\phi _m}}}\sqrt {\frac{x}{{\rho \left( {{a_m} - x{a_n}} \right)}}} } \right)dx.
\end{align}
 Applying Gauss-Chebyshev quadrature into above equation, we can obtain \eqref{the ergodic rate of user m under Rician fading channel}. The proof is completed.

\bibliographystyle{IEEEtran}
\bibliography{mybib}




%

\end{document}